%% file: main.tex
\documentclass[AMA,STIX1COL]{WileyNJD-v2}
\usepackage{moreverb}
\usepackage{amsmath,bm}
\usepackage{amsthm}
\usepackage{psfrag}
\usepackage{mathtools}
\usepackage{url}
\usepackage{lineno}
\usepackage{algorithm}
\usepackage{algpseudocode}
\usepackage{multicol}
\theoremstyle{definition}
\usepackage{subcaption}
\graphicspath{{figures/}}
\usepackage{xcolor}
\newcommand*{\rv}{} 
\newcommand*{\gv}{} 
\newcommand*{\gvc}{}

\newcommand\BibTeX{{\rmfamily B\kern-.05em \textsc{i\kern-.025em b}\kern-.08em
T\kern-.1667em\lower.7ex\hbox{E}\kern-.125emX}}

\articletype{}

\begin{document}

\title{Cellular Topology Optimization on Differentiable Voronoi Diagrams}

\author[1]{Fan Feng}
\author[1]{Shiying Xiong}
\author[1]{Ziyue Liu}
\author[1]{Zangyueyang Xian}
\author[2]{Yuqing Zhou}
\author[3]{Hiroki Kobayashi}
\author[3]{Atsushi Kawamoto}
\author[3]{Tsuyoshi Nomura*}
\author[1]{Bo Zhu}

\address[1]{\orgdiv{Computer Science Department}, \orgname{Dartmouth College}, \orgaddress{\state{New Hampshire}, \country{United States}}}
\address[2]{ \orgname{Toyota Research Institute of North America}, \orgaddress{\state{Michigan}, \country{United States}}}
\address[3]{\orgname{Toyota Central R \& D Labs}, \orgaddress{\state{Nagakute}, \country{Japan}}}

\corres{*Tsuyoshi Nomura, Koraku Mori Building 10F, 
1 - 4 - 14 Koraku, Bunkyo-ku, Tokyo 112-0004, Japan. 
\email{nomu2@mosk.tytlabs.co.jp}}


\abstract[Abstract]{
Cellular structures manifest their outstanding mechanical properties in many biological systems. One key challenge for designing and optimizing these geometrically complicated structures lies in devising an effective geometric representation to characterize the system's spatially varying cellular evolution driven by objective sensitivities. A conventional discrete cellular structure, e.g., a Voronoi diagram, whose representation relies on discrete Voronoi cells and faces, lacks its differentiability to facilitate large-scale, gradient-based topology optimizations. We propose a topology optimization algorithm based on a differentiable and generalized Voronoi representation that can evolve the cellular structure as a continuous field. The central piece of our method is a hybrid particle-grid representation to encode the previously discrete Voronoi diagram into a continuous density field defined in a Euclidean space. Based on this differentiable representation, we further extend it to tackle anisotropic cells, free boundaries, and 
functionally-graded cellular structures.
Our differentiable Voronoi diagram enables the integration of an effective cellular representation into the state-of-the-art topology optimization pipelines, which defines a novel design space for cellular structures to explore design options effectively that were impractical for previous approaches.
We showcase the efficacy of our approach by optimizing cellular structures with up to thousands of anisotropic cells, including femur bone and Odonata wing.
}

\keywords{Topology optimization; differentiable Voronoi diagram; cellular structure}


\maketitle

\input{intro}

\input{math}

\input{implementation}

\input{results}

\input{conclusion}

\section*{acknowledgements} 
This project was supported by Toyota Central R\&D Labs. The Dartmouth authors also acknowledge NSF-1919647, 2106733, 2144806, and 2153560. We credit the Houdini education license for producing the video animations.
\bibliography{refs}

\end{document}

%% file: intro.tex
\section{Introduction}
\label{sec:intro} 
Cellular and foam structures are ubiquitous in nature. Examples range from honeycombs \cite{qi2021advanced}, bone interiors \cite{gautam2021nondestructive}, and insect wings \cite{jongerius2010structural}, to foam bubbles \cite{kraynik2006structure}, all of which exhibit a wide range of material properties and functionalities that cannot be realized with traditional engineering configurations. Developing mechanistic understanding and numerical optimization of cellular structures with complex geometric patterns have received extensive attention in both scientific and engineering communities \cite{nazir2019state}. 
For example, experiments, simulations, and data analysis have been carried out to study the insect wing morphology \cite{salcedo2019computational}, leaf vein topology \cite{ma2021topology}, and sponge cake formation \cite{bousquieres2017functional}, for the sake of uncovering the structural principles and developmental mechanics underpinning these biological systems.
However, most of the existing studies and cellular structural designs are based on trial-and-error experiments \cite{bock2010generalized}, stochastic sampling \cite{kou2012microstructural,gostick2013random}, and procedural methods \cite{martinez2018polyhedral,lei2020parametric}, which heavily relied on the materials' homogenized properties and human engineers' prior expertise. First-principle methods guided by mathematically rigorous sensitivities and numerically efficient optimization frameworks have remained an unexplored field due to three interleaving challenges in tackling cellular structures' differentiation, generalization, and discretization. 

First, devising a differentiable representation for complex cellular structures is challenging. To the best of our knowledge, no current numerical approach can differentiate and optimize the topological evolution of a complex cellular structure on a single-cell level, i.e., to \textit{co-optimize the size, shape, and topology} of each local cell nested in the design domain by following continuous sensitivities of a specific design objective. Traditional cellular representations, 
with Voronoi diagrams \citep{voronoi1908nouvelles} as a particular example, represent the geometry and topology of a cellular pattern on a discrete mesh. 
By defining a distance field based on a set of site points and extracting the level set with an equal distance to the two closest points, a Voronoi diagram partitions a domain into a set of discrete cells, with each cell represented by a set of faces, which are borders between incident cells, and vertices, which are intersections of incident faces. On the simulation side, numerical solvers discretized on a Voronoi diagram, or its many variations (e.g., Power diagrams \cite{aurenhammer1987power}), have been used widely to facilitate adaptive simulations of large-scale fluid \cite{springel2010moving, ji2019new}, solid \cite{busto2020high}, and cosmology \cite{springel2010pur} systems in scientific computing. On the geometry processing side, discrete mesh generation algorithms for Voronoi diagrams, or its dual, Delauney tessellations, have been integrated into the state-of-the-art computational geometry libraries (e.g., CGAL \cite{cgal:k-vda2-22a}) to render support for various computational geometry applications. Recently, a new line of efforts has been devoted to developing parallel algorithms to generate large-scale Voronoi diagrams on modern GPUs \cite{ray2018meshless,liu2020parallel}. Despite the rapid advent 
of forward simulation and mesh generation, inverse procedures to differentiate a Voronoi diagram, i.e., to calculate the sensitivity of a Voronoi diagram's geometric shape and topological pattern with respect to its design variables, which are essential to accommodate the various structural design and optimization applications, has remained unexplored. In the state of the art, most of the computational design of Voronoi structures in a specific physical context relies on non-gradient or procedural algorithms for the design parameter exploration. 
Instances include \cite{lu2014build}, which optimizes the Voronoi structure by distributing the site points heuristically according to the structure's stress, \cite{oncel2019generation}, which initializes Voronoi site points inside high-density boundaries according to standard topology optimization results, and \cite{cucinotta2019stress}, which iteratively deletes and adds site points to obtain a desired stress distribution on the surface. None of the above methods optimizes a Voronoi structure based on continuous sensitivity analysis. 
An algorithmic strategy to probe the optimal solution, by structuring the design space or enforcing strong heuristics, becomes an imperative necessity to establish an effective numerical optimization paradigm for cellular structures. 

The second challenge for cellular structural optimization lies in the representation's generalizability. Most of the cellular patterns in nature do not rigorously satisfy the Voronoi definition. The cellular elements on these materials are \textit{\gv{geometrically} anisotropic}, \textit{spatially adaptive}, 
and \textit{with a free boundary}. 
First, cellular structures in nature are usually anisotropic, unlike a traditional Voronoi diagram whose cells exhibit a uniform distance metric in all directions. For example, the vein network of an insect wing partitions a thin shell into regions featured by a broad range of aspect ratios, encompassing cells that are small, large, thin, or slender \cite{hoffmann2018simple}. It has also been shown in the literature that anisotropic cellular structures can lead to lower structural compliance than isotropic ones \citep{ying2018anisotropic}. Representing an ensemble of anisotropic cellular features within a unified numerical framework requires an expressive geometric data representation that can characterize cellular patterns beyond a standard Voronoi diagram. 
Second, spatial adaptivity is essential for a cellular structure to obtain outstanding mechanical performance. 
The Centroidal Voronoi Tessellation (CVT) method \citep{du1999centroidal}, along with other functionally-graded lattice structures \cite{martinez2018polyhedral, do2021homogenization}, was invented to control the spatial distribution of Voronoi cells. 
Wu et al.~\citep{wu2017infill} utilized local volume constraints to design infill structures on a full scale, which formed uniform lattice structures. Non-uniform local volume constraints have been later proposed to design heterogeneous lattice structures~\cite{yi2019topology,schmidt2019structural}.
Lastly, numerous cellular structures in nature form circular boundaries by themselves. But the regular Voronoi partition spans the entire domain. Other cellular-related representations such as the open-cell porous structure have also been investigated~\cite{martinez2016procedural, tian2020organic}. However, there is no existing work that enables the representation of cellular structure with free boundaries.


Third, it is challenging to devise an efficient numerical discretization to support cellular structure simulation and optimization.
The existing discrete data structures underpinning topology optimization applications are composed mainly of grids, particles/meshes, and their hybrid. 
The lattice grid is the most popular data structure empowering density-based topology optimization~\cite{bendsoe2003topology}. 
However, this approach requires a very high grid resolution and well-designed filters to emerge local fine features \cite{wu2015system, liu2018narrow} 
Homogenization-based topology optimization plus an additional dehomogenization-based post-processing step have been applied to design cellular structures in a multi-scale way~\cite{wu2019design, groen2020homogenization, telgen2022topology}. The resulting optimized designs, however, are considered near-optimal due to the reconstruction errors from dehomogenization. It is noted that a similar framework has also been applied to the optimal design of microreactors with space-filling microchannel flow fields~\cite{zhou2022inverse}. By simultaneously optimizing micro-structural cell geometries and their distributions on the macro scale, concurrent two-scale topology optimization can be used to design lattice structures~\cite{zhu2017two,deng2017concurrent,wang2019concurrent}. 
Besides grid-based methods, Lagrangian representations, manifesting mainly as meshless particles \citep{li2021lagrangian} and simplicial meshes \cite{cho2006topology} play an important role in algorithms that co-optimize a target structure's shape and topology. In these approaches, the simulation is discretized directly on a mesh or particles, and the optimization is enabled by moving the mesh vertices or meshless points in the design domain according to their sensitivities. The remeshing overhead, including repairing the simplicial mesh or redistributing particles, is one of the main bottlenecks of these approaches. None of the existing Lagrangian methods can handle the optimization of complex cellular structures.
In addition to grid-based and particle-based approaches, hybrid schemes were also explored extensively in topology optimization. For instance, the moving morphing component method (MMC)~\cite{guo2014doing, liu2017additive} and the geometric projection method~\cite{norato2015geometry, kazemi2020multi} optimize the geometric parameters (e.g., position and orientation) of a set of moving primitives that characterize the structure by projecting the union of these components on a background grid. Similar ideas can be seen in Li et al.\citep{li2021lagrangian}, which utilizes the communication between fixed material-point quadratures and moving particles to achieve a sub-cell representation for complex structures. Motivated by these approaches, we propose a hybrid grid-particle representation to encode the cellular pattern.
Our approach distinguishes itself from the previous ones by co-optimizing the particles' positions and their local distance metric to describe the local cellular geometry, which enabled our method to evolve complicated structures without relying on abundant point samples or prohibitively high-resolution grids.



To address the aforementioned three challenges, we propose a computational approach that can differentiate, generalize, and explore the design space of complex cellular structures. 
We first formulate a differentiable Voronoi representation, inspired by the Soft Voronoi representation in the recent computer vision literature \citep{williams2020voronoinet}, where they devise a differentiable Voronoi diagram and embed it into a generative deep network for solid geometry representation.  
Akin to their method, we leverage a softmax function to partition the Voronoi regions in a differentiable fashion. We further devised a novel filter to extract the density-based Voronoi faces, which are essential for structural optimization, and devise a differentiable scheme to evolve their topology with respect to the design variables.
Next, we developed a generalized Voronoi representation for anisotropic, spatially adaptive, and free-boundary-enabled cellular structures. By adopting Mahalanobis metric tensor \citep{richter2015mahalanobis} on each Voronoi site point, our method enables a heterogeneous and anisotropic Voronoi tessellation. By letting each point carry its metric tensor, the anisotropy of the Voronoi cells can vary vigorously over the design space. 
These localized metric tensors, in conjunction with the moving particles, enable a spatially adaptive and anisotropic representation of cellular patterns, which greatly extend the design space of the original Voronoi diagram. To represent structures with a free boundary, we further modify our representation to support a foamy boundary, which enhances the traditional Voronoi diagram which functions primarily as a spatial partition. 
Last, with the hybrid grid-particle discretization, our method incorporates the differentiable Voronoi representation into the conventional Solid Isotropic Material with Penalization method (SIMP) topology-optimization framework. Furthermore, we implement a k-nearest neighbor search method to localize the computation of the differentiable Voronoi diagram to reduce the computational cost of the Voronoi diagram from $O(N^2)$ to $O(k\log N)$.

The rest of the paper is organized as follows. We discuss the mathematical formulation of the differentiable, anisotropic, and spatially adaptive Voronoi partition with a free boundary in Section \ref{sec:diff_voronoi}. We then integrate the Voronoi formulation into the density-based topology optimization framework in Section \ref{sec:topo}. The overall algorithm workflow and optimization are presented in Section \ref{sec:implementation}. Section \ref{sec:results} presents extensive numerical results before conclusions are drawn in Section \ref{sec:discussion}.

%% file: math.tex
\section{Differentiable Voronoi diagram}
In this section, we will formulate the mathematical definition for our differentiable, anisotropic, and spatially adaptive Voronoi representation with free boundary and its sensitivity analysis.

\label{sec:diff_voronoi}
\begin{figure}
     \centering
     \begin{subfigure}{0.28\textwidth}
         \includegraphics[width=\textwidth]{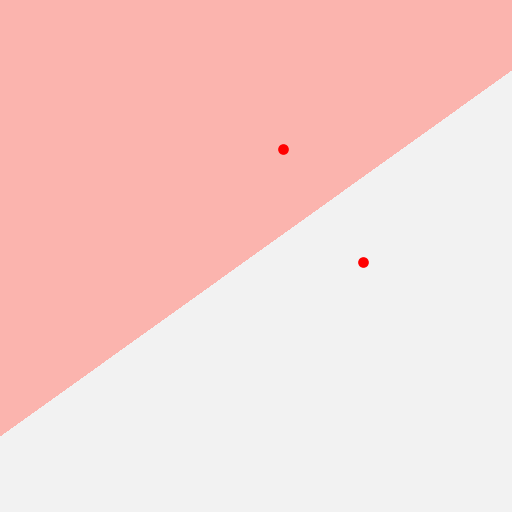}
         \caption{Standard Voronoi tessellation.\\}
         \label{fig:2_pts_voronoi_patch}
     \end{subfigure}
     \hfill
     \begin{subfigure}{0.42\textwidth}
         \includegraphics[width=\textwidth]{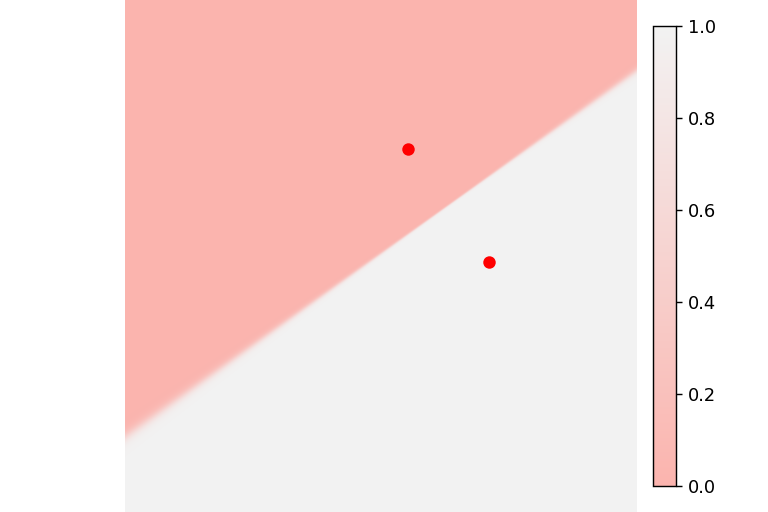}
         \caption{Voronoi tessellation using the Softmax function.\\}
         \label{fig:2_pts_softmax}
     \end{subfigure}
     \hfill
     \begin{subfigure}{0.28\textwidth}
         \includegraphics[width=\textwidth]{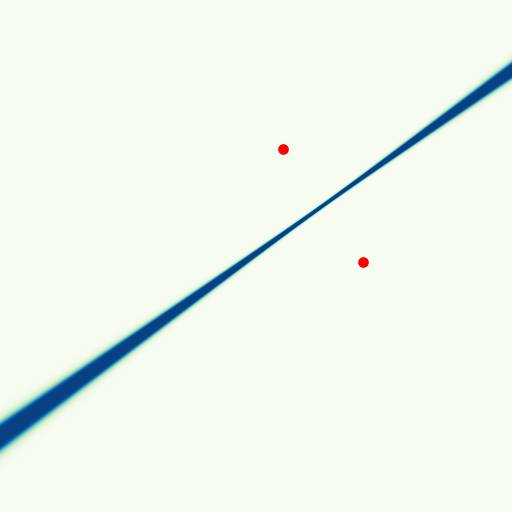}
         \caption{The extracted Voronoi edge between two regions ($\rho$).}
         \label{fig:2_pts_diff}
    \end{subfigure}
    \caption{Visualization of different quantities formed with two Voronoi points. The red dots indicate the positions of the control points.}
    \label{fig:2_pts_voronoi}
\end{figure}

\begin{figure}
     \centering
     \begin{subfigure}[b]{0.32\textwidth}
         \centering
         \includegraphics[width=\textwidth]{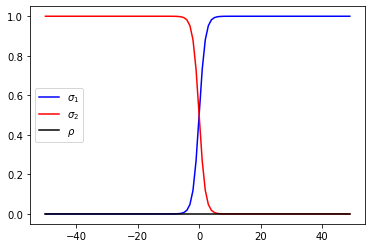}
         \caption{$\beta=1$}
         \label{fig:beta_1}
     \end{subfigure}
     \hfill
     \begin{subfigure}[b]{0.32\textwidth}
         \centering
         \includegraphics[width=\textwidth]{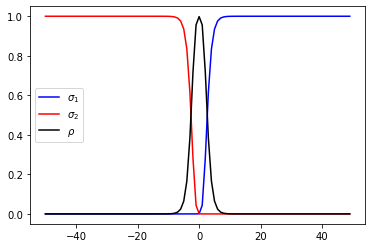}
         \caption{$\beta=10$}
         \label{fig:beta_10}
     \end{subfigure}
     \hfill
     \begin{subfigure}[b]{0.32\textwidth}
         \centering
         \includegraphics[width=\textwidth]{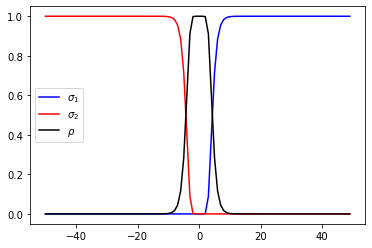}
         \caption{$\beta=50$}
         \label{fig:beta_50}
    \end{subfigure}
    \hfill
    \caption{Visualization of the softmax function $\sigma_1$ (blue), $\sigma_2$ (red) and $\rho$ (black) in one dimension with different $\beta$ values.}
    \label{fig:beta}
\end{figure}
\subsection{Differentiable Voronoi diagram}
Given a set of site points $P = \{\bm x_m~|~m = 1,2,\cdots,N^c\}$ in $N^d$ dimensional domain $\Omega$, a Voronoi diagram is a partition of $\Omega$ into $N^c$ polygonal regions, called Voronoi cells, with each Voronoi cell $V_m$ associated with a site point $\bm x_m$. Each Voronoi cell manifests a geometric property that $\bm x_m$ is the closest of the given points to any $\bm x \in V_m$ \cite{shamos1975closest}. In other words, if $\textrm{d}_m (\bm x,\bm x_m)$ denotes the norm between the points $\bm x$ and $\bm x_m$, then 
\begin{equation}
    V_m =\{\bm x\in \Omega~|~\textrm{d}_m (\bm x,\bm x_m)\leq \textrm{d}_n (\bm x,\bm x_n)~~\textrm{for all}~~ n = 1,2,\cdots, N^c\}.
\end{equation}

We remark that the norm $\textrm{d}_m$ can take different forms. If $\textrm{d}_m$ is a Euclidean norm, all Voronoi cells are bounded by some straight segments; otherwise, the boundary of some Voronoi cells may be curved segments (non-Euclidean norms will be discussed in Section \ref{sec:aniso_voronoi}).
Based on the Euclidean norm, we show in Figure (\ref{fig:2_pts_voronoi_patch}) the partition of a two-dimensional domain into two Voronoi cells.

To enable a differentiable partition, we introduce a modified Voronoi partition based on the softmax function \cite{boltzmann1868studien} as follows:
\begin{equation}
    S_m(\bm x) = f_m^{N^c}\left[e^{-\textrm{d}_1(\bm x,\bm x_1)},\cdots,e^{-\textrm{d}_{N^c}(\bm x,\bm x_{\rv{N^c}})}\right], ~~ m\in I_v,
    \label{eq:s_def}
\end{equation}
where $f_m^{N^c}$ is a fractional equation
\begin{equation}
f^{N^c}_m(x_1,\cdots,x_{N^c}) = \frac{x_m}{\sum_{n\in I_v}x_n}, 
\end{equation}
and $I_v$ is an index set with
\begin{equation}
   I_v\equiv\{1,\cdots,N^c\}.
   \label{eq:I_v}
\end{equation}
The softmax function is originated from physics and statistical mechanics, which recently has also been widely used as an activation function in modern neural networks \cite{bridle1990probabilistic}. The softmax function selects one item from a group by following some distance metric in a ``soft" (and therefore differentiable) manner. In other words, we can build a program without an ``if" statement to calculate the 0-1 density distribution that specifies the region of a given Voronoi point. Figure (\ref{fig:2_pts_softmax}) shows the Voronoi partition formed by the softmax function, where the gray region has a value of one and the red region has a value of zero to indicate the coverage of the lower point. 

In our modified Voronoi partition, we calculate the density distribution for a specific site point in $\Omega$ as
\begin{equation}
\rho (\bm x) =1- \sum_{m\in I_v} \left[S_m(\bm x)\right]^{\beta},
\label{eq:rho}
\end{equation}
where $\beta$ controls the sharpness of the edge as illustrated in Figure \ref{fig:beta}. As shown in Figure (\ref{fig:beta_1}), when $\beta=1$, no edge will appear since the softmax values are simply added together. Comparing Figure (\ref{fig:beta_10}) and Figure (\ref{fig:beta_50}), we can see there is a positive correlation between the value of $\beta$ and the sharpness of the edges. \rv{It is noticeable that as $\beta$ increases, the thickness of edges also increases. We did not treat $\beta$ as an optimizing variable because there is a more direct relationship between the edge thickness and metric tensor, as discussed in the next subsection.}

\begin{figure}
     \centering
     \begin{subfigure}{0.3\textwidth}
         \includegraphics[width=\textwidth]{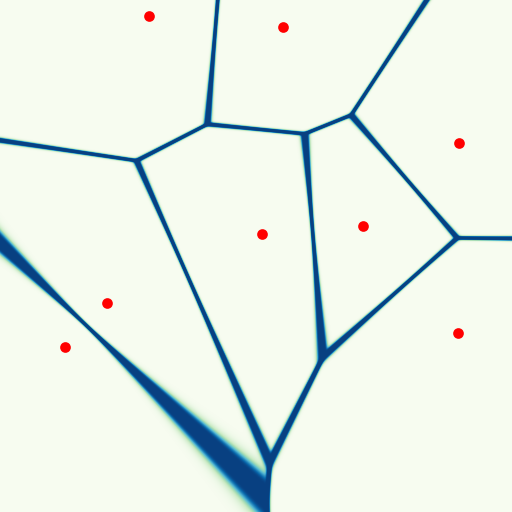}
         \caption{Isotropic Voronoi diagram with \\ $\bm D_m=[500,0\rv{;~}0,500]$. \\}
         \label{fig:8_pts_isotropic}
     \end{subfigure}
     \hfill
     \begin{subfigure}{0.3\textwidth}
         \includegraphics[width=\textwidth]{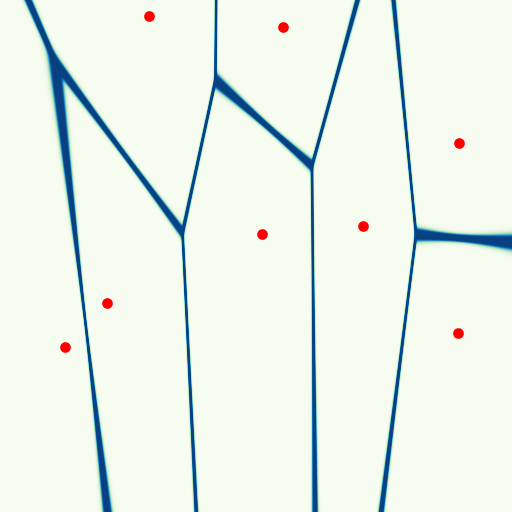}
         \caption{Anisotropic Voronoi diagram with \\ $\bm D_m=[750,0\rv{;~}0,250]$. \\}
         \label{fig:8_pts_anisotropic}
     \end{subfigure}
     \hfill
     \begin{subfigure}{0.3\textwidth}
         \includegraphics[width=\textwidth]{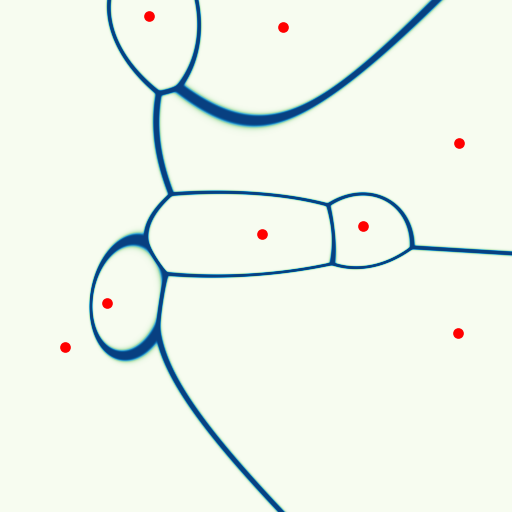}
         \caption{Heterogeneously Anisotropic Voronoi diagram with diagonal entries ranging $[0,1000]$ and off-diagonal entries ranging $[-50,50]$. }
         \label{fig:8_pts_anisotropic_het}
    \end{subfigure}
    \caption{Illustration of the anisotropic Voronoi diagram. Domain is a $1\times1$ square.}
    \label{fig:anisotropy}
\end{figure}

We remark that the density distribution in Eq.~\eqref{eq:rho} defined based on softmax functions in Eq.~\eqref{eq:s_def} can subdivide the domain into approximate Voronoi cells in which the cell boundaries have a finite thickness as shown in Figure \eqref{fig:2_pts_diff}. 

\subsection{Anisotropic Voronoi structure}
\label{sec:aniso_voronoi}
We enhance the local geometric expressiveness by defining the distance metric between points based on Mahalanobis distance \cite{mahalanobis1936generalized}
\begin{equation}
    \textrm{d}_m(\bm x, \bm x_m )= \sqrt{(\bm x-\bm x_m)^T \bm A_m (\bm x-\bm x_m)},
    \label{eq:dmxm}
\end{equation}
where 
\begin{equation}
   \bm A_m = \bm D_m \bm D_m^T,
   \label{eq:Am}
\end{equation}
$\bm D_m$ is a $N^d$-by-$N^d$ symmetric matrix. This construction ensures that $\bm A_m$ is a positive definite symmetric matrix. 

Figure \ref{fig:anisotropy} illustrates the difference between isotropic Voronoi tessellation (Figure \ref{fig:8_pts_isotropic}), anisotropic Voronoi tessellation (Figure \ref{fig:8_pts_anisotropic}) where metric tensors are the same for each control point, and heterogeneous anisotropic Voronoi tessellation (Figure \ref{fig:8_pts_anisotropic_het}) where metric tensors are different for each control point. Curved edges only form when we have spatially varying metric tensors. \gv{It is worth noticing that we use "anisotropic" to feature the geometry of the Voronoi structure characteristics, which only refers to how we define the Voronoi distance, in contrast to the material anisotropy used in continuum mechanics.}

\begin{figure}
    \begin{subfigure}{0.3\textwidth}
         \includegraphics[width=\textwidth]{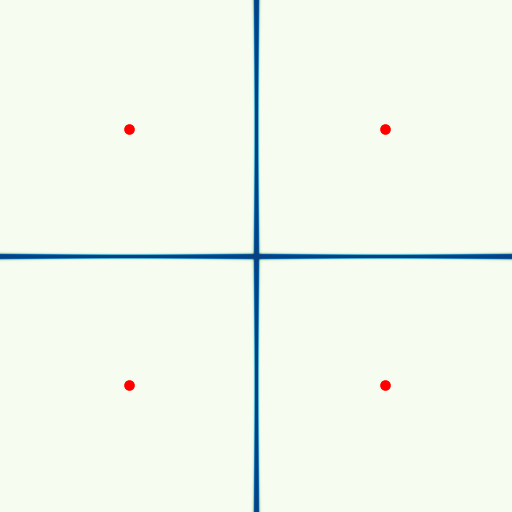}
         \caption{Voronoi cells Without free boundaries, equivalent to $\epsilon_s=0$ with $\bm D_m=500 \bm I$.\\ \newline}
         \label{fig:fb}
    \end{subfigure}
    \hfill
    \begin{subfigure}{0.3\textwidth}
        \includegraphics[width=\textwidth]{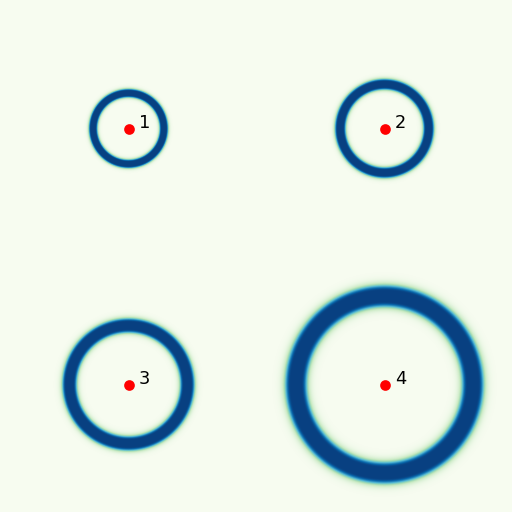}
        \caption{Free boundary cells in different sizes with $\bm D_1=500\bm I$, $\bm D_2=400\bm I$, $\bm D_3=300\bm I$, $\bm D_4=200\bm I$ respectively, and $\epsilon_s=1e-15$.}
        \label{fig:fb_sizes}
    \end{subfigure}
    \hfill
    \begin{subfigure}{0.3\textwidth}
         \includegraphics[width=\textwidth]{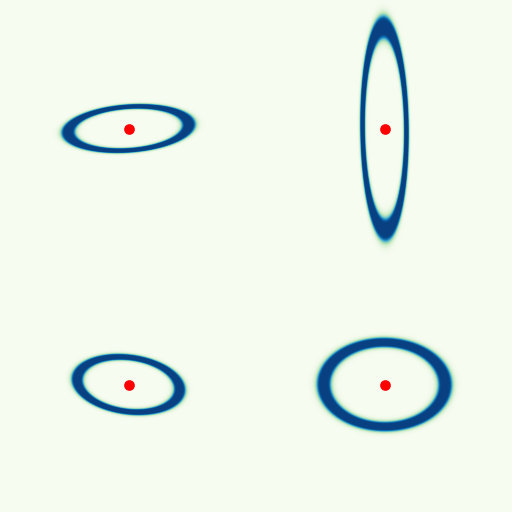}
         \caption{Free boundary cells with $\bm D_m$ diagonal entries ranging $[0,1000]$ and off-diagonal entries ranging $[-50,50]$.  \newline}
         \label{fig:fb_aniso_het}
    \end{subfigure}
    
    \caption{Illustration of Voronoi diagram with free boundary.}
    \label{fig:free_boundary}
\end{figure}

\subsection{Voronoi structure with free boundary}
\label{sec:fb_voronoi}
We introduce a positive constant $\epsilon_s$ to characterize the cellular structure with a free boundary. The key idea is to introduce a universal virtual point that is outside all the current Voronoi points to help define a Voronoi face that is on the free boundary. In particular, we redefine $I_v$ in \eqref{eq:I_v} and $S_m$ in \eqref{eq:s_def} as
\begin{equation}
    I_v = \{0,\cdots,N^c\}
    \label{eq:Iv_0}
\end{equation}
and 
\begin{equation}
    S_m(\bm x) = f_m^{N^c+1}\left[\epsilon_s,e^{-\textrm{d}_1(\bm x,\bm x_1)},\cdots,e^{-\textrm{d}_{N^c}(\bm x,\bm x_{\rv{N^c}})}\right], ~~ m\in I_v.
    \label{eq:Sm}
\end{equation}
Here, the definition of 
\begin{equation}
    S_{0}(\bm x) = \frac{\epsilon_s}{\sum_{n=1}^{N^c} e^{-\textrm{d}_n(\bm x,\bm x_n)}+ \epsilon_s},
    \label{eq:S0}
\end{equation}
and
\begin{equation}
    S_{m}(\bm x) = \frac{e^{-\textrm{d}_m(\bm x,\bm x_m)}}{\sum_{n=1}^{N^c} e^{-\textrm{d}_n(\bm x,\bm x_n)}+ \epsilon_s},
    \label{eq:Sm_fb}
\end{equation}
ensures that $\rho$ in \eqref{eq:rho} has a compact support. 

The effect of free boundary and how different metric tensors influence the sizes and shapes of the cells are illustrated in Figure \ref{fig:free_boundary}. Figure \ref{fig:fb}) shows the differentiable Voronoi diagram without a free boundary. Figure \ref{fig:fb_sizes}) illustrates the points with different metric tensors $D_m$ forming different sizes of free boundaries. Figure \ref{fig:fb_aniso_het}) shows cells in different shapes are formed by also using different $D_m$ on each control point.

\subsection{Differentiation of the modified Voronoi partition}
The distribution $\rho$ in \eqref{eq:rho} with \eqref{eq:dmxm}, \eqref{eq:Iv_0}, and \eqref{eq:Sm} is determined by positions $\bm x_m, m = 1,2,\cdots,N^c$ of the Voronoi centers, matrices $\bm D_m, m = 1,2,\cdots,N^c$ associated with metrics in \eqref{eq:Am}, and one constant $\epsilon_s$. These parameters can generate material distributions with different geometric properties and therefore can be used as design variables that feature a cellular structure. For computing partial derivatives of $\rho$ with respect to parameters $\bm x_m$ and $\bm D_m$, we summarize the results as the following theorem:
\begin{theorem}
Assuming $\rho$ is defined by \eqref{eq:rho} with \eqref{eq:dmxm}, \eqref{eq:Iv_0}, and \eqref{eq:Sm}, its partial derivatives with respect to parameters $\bm x_m, \bm D_m$ can be calculated as
\begin{equation}
    \begin{dcases}
     \frac{\partial \rho(\bm x) }{\partial \bm x_n}
    = -\sum_{m=0}^{N^c}\beta \left[S_m(\bm x)\right]^{\beta-1}\frac{\partial  S_m (\bm x)}{\partial \bm x_n},~n=1,2,\cdots,N^c\\
    \frac{\partial \rho (\bm x)}{\partial \bm D_n}
    = -\sum_{m=0}^{N^c}\beta \left[S_m(\bm x)\right]^{\beta-1}\frac{\partial  S_m (\bm x) }{\partial \bm D_n},~n=1,2,\cdots,N^c\\
    \end{dcases}
    \label{eq:grad_rho}
\end{equation}
where $\forall m,n=1,2,\cdots,N^c,$
\begin{equation}
    \begin{dcases}
    \frac{\partial S_{m}(\bm x)}{\partial \bm x_n}
    = S_m(\bm x)[S_n(\bm x)-\delta_{mn}]\bm Y_n(\bm x),\\
    \frac{\partial S_{0}(\bm x)}{\partial \bm x_n}
    =  S_{0}(\bm x)S_n(\bm x)\bm Y_n(\bm x),\\
    \frac{\partial S_{m}(\bm x)}{\partial \bm D_n}
    =S_m(\bm x)[S_n(\bm x)-\delta_{mn}] \textrm{d}_n(\bm x, \bm x_n)\left[\bm X_n (\bm x)\otimes \bm X_n(\bm x)\right]\bm D_n,\\
     \frac{\partial S_{0}(\bm x)}{\partial \bm D_n}
    =S_{0}(\bm x)S_n(\bm x) \textrm{d}_n(\bm x, \bm x_n)\left[\bm X_n (\bm x)\otimes \bm X_n(\bm x)\right]\bm D_n.\\
    \end{dcases}
    \label{eq:dSm}
\end{equation}
Here,
\begin{equation}
\begin{dcases}
\bm X_n(\bm x)=\frac{\bm x_n- \bm x}{\textrm{d}_n(\bm x, \bm x_n)},~n=1,2,\cdots,N^c,\\
    \bm Y_n (\bm x) = \bm A_n(\bm x)\bm X_n(\bm x) ,~n=1,2,\cdots,N^c,\\
\end{dcases}
\label{eq:XY}
\end{equation}
and $\delta_{mn}$ is the Kronecker delta function
\begin{equation}
    \delta_{mn}=\begin{dcases}
    1,~~~~m=n,\\
    0,~~~~m\neq n.
    \end{dcases}
\end{equation}
\end{theorem}

\begin{proof}
Equations \eqref{eq:grad_rho} can be derived directly from the chain rule and \eqref{eq:rho}. Next, we prove \eqref{eq:dSm}.
From \eqref{eq:Sm}, we have $\forall m,l=1,2,\cdots,N^{c}$
\begin{equation}
    \begin{dcases}
\frac{\partial S_m}{\partial \textrm{d}_l(\bm x, \bm x_l)}
=S_m\left(S_n-\delta_{ml}\right),\\
\frac{\partial S_0}{\partial \textrm{d}_l(\bm x, \bm x_l)}
=S_0S_l.\\
    \end{dcases}
    \label{eq:dSm0}
\end{equation}
Taking the partial derivatives of \eqref{eq:dmxm} yields
\begin{equation}
\begin{dcases}
 \frac{\partial \textrm{d}_l(\bm x, \bm x_l)}{\partial \bm x_n}=\delta_{ln} \bm Y_n,\\
 \frac{\partial\textrm{d}_l(\bm x, \bm x_l)}{\partial \bm D_n} = \delta_{ln} \textrm{d}_n(\bm x, \bm x_n) [\bm X_n(\bm x) \otimes \bm X_n(\bm x)]\bm D_n.
\end{dcases}
\label{eq:ddl}
\end{equation}
Combining \eqref{eq:dSm0} and \eqref{eq:ddl} using the chain rule yeilds \eqref{eq:dSm}.
\end{proof}

\section{Topology optimization}
\label{sec:topo}

\subsection{Incorporating Voronoi to topology optimization} \label{subsec:c_topo}
In topology optimization, we store the density and projected density distributions $\rho$ on a structured Cartesian grid in domain $\Omega$, the density and projected density distributions at each grid point are denoted as $\rho_e, e=1,\cdots,N^e$. 
Without causing ambiguity, we also denote $\rho$ and as vectors consisting of all discrete $\rho_e$ respectively.
Following \cite{ferrari2020new}, we write a topology optimization problem as the minimization of the mean compliance $c(\bm x, \bm D)$ for the filtered density distribution $\Tilde{\rho}$ in the nested formulation
\begin{equation}
    \mathop{\rm minimize}_{\bm x, \bm D} \quad c(\bm x, \bm D) := \bm F^T \bm U \quad
    \textrm{subject to} \quad \frac{V(\Tilde{\rho})}{V(\Omega)}\leq \bar{f},
\label{eq:topo}
\end{equation}
where $\Tilde{\rho}=f_H(\rho)\equiv\frac{\tanh(\gamma\eta)+\tanh(\gamma(\rho-\eta))}{\tanh(\gamma\eta)+\tanh(\gamma(1-\eta))}$ is  the relaxed Heaviside projection \cite{wang2011projection} with threshold $\eta=0.5$ and $\gamma$ starting from 1 doubled every 50 frames, gradually producing a more binary filtered density field. $\bm{F}$ is the external force vector, $\bm{U}$ is the global displacement vector, $V(\Tilde{\rho})$ is the material volume, $V(\Omega)$ is the volume of the design domain $\Omega$, and $\bar{f}$ is the prescribed volume fraction. Design variables are the positions $\bm x$ and metric tensors $\bm D$ of the site points. $\bm x$ takes a range \rv{at least three times as large as the design domain on each axis} (i.e., Voronoi points can be positioned outside of the design domain \rv{and not form density inside, thus not influencing the final structure}). The diagonal entries in $\bm D$ take a range as large as $[0,2000]$, where a larger range produces a more variation in edge thickness, and the off-diagonal terms in $\bm D$ have the range $[-100,100]$. The global displacement vector $\bm{U}$ is obtained by solving the following force equilibrium
\begin{equation}
    \bm K(\bm x, \bm D)\bm U=\bm F,
\end{equation}
where $\bm K$ is the global stiffness matrix.

Following the modified Solid Isotropic Material with Penalization (SIMP) method, the stiffness matrix of each grid element is interpolated as follows,
\begin{equation}
\begin{aligned}
    E_e&=E_{\mathrm{min}}+\Tilde{\rho}_e^p(E_0-E_{\mathrm{min}}),\\
    \bm k_e&=E_e\bm k_0,
\end{aligned}
\label{eq:topo_stiffness}
\end{equation}
where $E_e$ is the interpolated Young's modulus of each element, $E_0$ is Young's modulus of a completely solid element, $E_{\mathrm{min}}$ is a very small value to prevent singularity of the global stiffness matrix \gv{and $\bm k_0$ is the element stiffness matrix given by $\bm k_0=\int_{V} \bm B^T \bm E \bm B \,dV$ with Young's modulus equal to 1. In the analytical expression of $\bm k_0$, $V$ is the volume of one element, $\bm B$ is the strain displacement Matrix and $\bm E$ is the stress-strain matrix. One can form the element stiffness matrix through Gaussian quadrature discretization or a pre-calculated matrix. We refer readers to \cite{sigmund200199} and \cite{liu2014efficient} for the formulation of the element stiffness matrix in 2D and 3D respectively. We assemble $\bm K$ as a sparse matrix by filling in the entries of nodes associated with each element stiffness matrix.} Here we choose the penalizing power $p=1$ because our Voronoi formulation forms relatively binary pattern by its definition.

\subsection{Sensitivity analysis}
The sensitivity analysis is given as followings using chain rule
\begin{equation}
\begin{dcases}
    \frac{\partial c}{\partial \bm x_m}=\sum_e\frac{\partial{c}}{\partial \Tilde{\rho}_e}\frac{\partial{\Tilde{\rho}_e}}{\partial \rho_e}\frac{\partial \rho_e}{\partial \bm x_m},\\
    \frac{\partial c}{\partial \bm D_m}=\sum_e\frac{\partial{c}}{\partial \Tilde{\rho}_e}\frac{\partial{\Tilde{\rho}_e}}{\partial \rho_e}\frac{\partial \rho_e}{\partial \bm D_m},\\
    \frac{\partial V}{\partial \bm x_m}=\sum_e\frac{\partial V}{\partial \Tilde{\rho}_e}\frac{\partial \Tilde{\rho}_e}{\partial \rho_e}\frac{\partial \rho_e}{\partial \bm x_m},\\
    \frac{\partial V}{\partial \bm D_m}=\sum_e\frac{\partial V}{\partial \Tilde{\rho}_e}\frac{\partial \Tilde{\rho}_e}{\partial \rho_e}\frac{\partial \rho_e}{\partial \bm D_m},\\
\end{dcases}
\label{eq:d_all}
\end{equation}
where
\begin{equation}
\begin{dcases}
    \frac{\partial c}{\partial \Tilde{\rho}_e}=-p\Tilde{\rho}_e^{p-1}(E_0-E_{\mathrm{min}})u_e^T k_0 u_e,\\
    \frac{\partial \Tilde{\rho}_e}{\rho_e}=\gamma\frac{1-\tanh(\gamma (\rho_e-\eta))^2}{\tanh(\gamma \eta)+\tanh(\gamma(1-\eta))},
\end{dcases}
\end{equation}
and $\partial \rho_e/\partial \bm x_m,\partial \rho_e/\partial \bm D_m$ are defined as \eqref{eq:grad_rho}.




%% file: implementation.tex
\begin{algorithm}
\caption{Calculation of differentiable Voronoi diagram}
\textbf{Input:} Positions Voronoi points $\bm x$, Metric tensor $\bm D$\\
\textbf{Output:} Neighbors field $nbs$, Density field $\rho$
\begin{algorithmic}[1]
\For{$i \gets 0$ to $N^c$}  \Comment{update $\bm A$}
\State $\bm A[i]=\bm D[i]@\bm D[i].T$ 
\EndFor
\For{$i \gets 0$ to $N^e$} \Comment{update neighbors}
\State $nbs[i] \gets$ Find\_K\_Nearest\_Nb($\bm x_i$, $\bm x$)  \Comment{stores the index of the i\textsuperscript{th} neighbor}
\EndFor
\For{$i \gets 0$ to $N^e$} \Comment{update denominator in Eq. \ref{eq:Sm_fb}}
    \For{$j \gets 0$ to $k$}                    
        \State $Sum$ $\gets$ {$Sum + \exp(-d(\bm x_i,\bm x_{nbs[i][j]}))$}
    \EndFor
    \State $softmax\_sum[i]\gets Sum + \epsilon_s$ \Comment{store the intermediate softmax sum in a field}
\EndFor

\For{$i \gets 0$ to $N^e$}  \Comment{update $\rho$}
    \For{$j \gets 0$ to $k$}
        \State $\rho_i \gets 1-\left[\exp \left(\frac{-d(\bm x_i,\bm x_{nbs[i][j]})}{softmax\_sum[i]}\right)\right]^{\beta}$ 
    \EndFor
\EndFor
\end{algorithmic}
\label{alg:density}
\end{algorithm}

\begin{algorithm}
\gvc
\caption{\gvc Calculation of $\partial \rho / \partial \bm x$}
\textbf{Input:}field of int array $nbs$, field  $\frac{\partial{c}}{\partial \Tilde{\rho}_e}$ and $\frac{\partial{\Tilde{\rho}_e}}{\partial \rho_e}$\\
\textbf{Output:} derivative of compliance with respect to Voronoi point positions $\frac{\partial c}{\partial \bm x}$
\begin{algorithmic}[1]
\For{$e \gets 0$ to $N^e$} \Comment{calculate $\frac{\partial \rho_e(\bm x)}{\partial \bm x_{nbs[e][m]}}$, parallelizable}
\For{$m \gets 0$ to $k$}
\State $\frac{\partial \rho_e(\bm x)}{\partial \bm x_{nbs[e][m]}}\gets \bm 0$
\For{$n \gets 0$ to $k$}
\State $\frac{\partial S_{nbs[e][m]}(\bm x)}{\partial \bm x_n} \gets S_{nbs[e][m]}(\bm x)[S_{nbs[e][n]}(\bm x)-\delta_{(nbs[e][m])(nbs[e][n])}]\bm Y_{nbs[e][n]}(\bm x)$ \Comment{Eq. \ref{eq:dSm}}
\State $\frac{\partial \rho_e(\bm x)}{\partial \bm x_{nbs[e][m]}} \mathrel{-}= \beta \left[S_{nbs[e][m]}(\bm x)\right]^{\beta-1}\frac{\partial S_{nbs[e][m]}}{\partial \bm x_{nbs[e][n]}}$ \Comment{Eq. \ref{eq:grad_rho}}
\EndFor
\EndFor
\EndFor

\For{$e \gets 0$ to $N^e$} \Comment{calculate $\frac{\partial c}{\partial \bm x_{nbs[e][m]}}$, not parallelizable}
\For{$m \gets 0$ to $k$}
\State $\frac{\partial c}{\partial \bm x_{nbs[e][m]}}\mathrel{+}=\frac{\partial{c}}{\partial \Tilde{\rho}_e}\frac{\partial{\Tilde{\rho}_e}}{\partial \rho_e}\frac{\partial \rho_e}{\partial \bm x_{nbs[e][m]}}$ \Comment{\ref{eq:d_all}}
\EndFor
\EndFor
\end{algorithmic}
\label{alg:dc_dx}
\end{algorithm}

\begin{algorithm}
\caption{Topology optimization with differentiable Voronoi diagrams}
\textbf{Input:} Volume fraction $f$, Number of Voronoi points $N^c$, Metric tensor $\bm D$\\
\textbf{Output:} Density field $\Tilde{\rho}$
\begin{algorithmic}[1]
\State Initialize positions of Voronoi points $\bm x$ randomly in a coarse grid
\State Iteration counter $i=0$
\State Design variable change $\Delta=1$
\State Projection parameter $\beta=1$
\While{$i<i_{max}$ and $\delta>1e-4$}
\State Calculate $\rho$ \Comment{Eq. \ref{eq:rho}, Algorithm \ref{alg:density}}
\State Density projection $\Tilde{\rho} \leftarrow \rho$
\State Compute compliance $c$   \Comment{Eq. \ref{eq:topo} and Eq. \ref{eq:topo_stiffness}}
\State Calculate $\frac{\partial c}{\partial \bm D}$,$\frac{\partial c}{\partial \bm x}$,$\frac{\partial V}{\partial \bm x}$,$\frac{\partial V}{\partial \bm D}$    \Comment{Eq. \ref{eq:d_all}, Algorithm \ref{alg:dc_dx}}
\State Update $\bm D$ and $\bm x$ with MMA optimizer
\State $\Delta=\max\{\max\limits_{\forall m} (|\bm D_m^i-\bm D_m^{i-1}|), \max\limits_{\forall m} (|\bm x_m^i-\bm x_m^{i-1}|)\} $
\State $i=i+1$
\EndWhile
\end{algorithmic}
\label{alg:overall}
\end{algorithm}

\section{Implementation}
\label{sec:implementation}

\paragraph{Neighbor Search}
\begin{figure}
    \centering
    \includegraphics[width=0.45\textwidth]{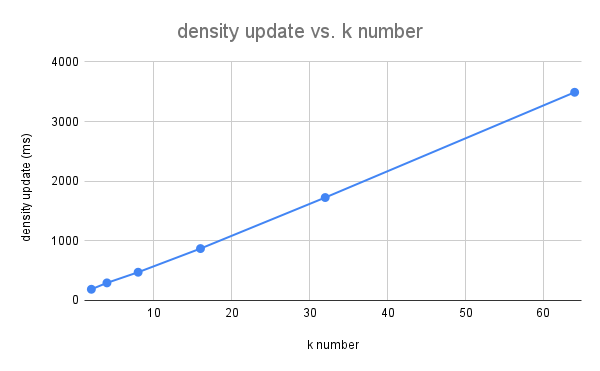}
    \includegraphics[width=0.45\textwidth]{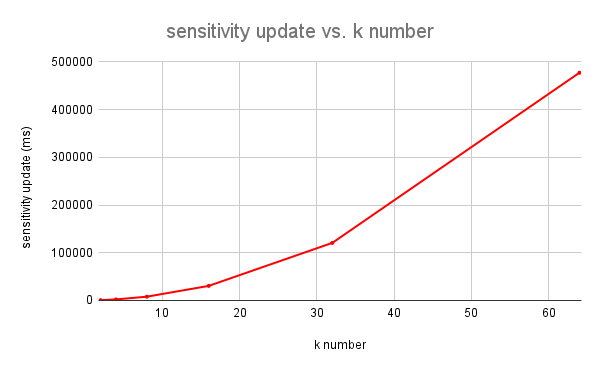}
    \caption{Computational cost density update and sensitivity update VS. $k$ number in a $1024x1024$ grid with 256 Voronoi points.}
    \label{fig:k_statistics}
\end{figure}
The computation cost is high for calculating the gradients with respect to $\bm x$ and $\bm D$, since it involves double loops as shown in (\ref{eq:grad_rho})(\ref{eq:dSm}) as both $m$ and $n$ need to iterate from 1 to $N^c$. The computation cost is therefore $O((N^c)^2)$ for each $\frac{\partial \rho_e}{\partial \bm x}$ and $\frac{\partial \rho_e}{\partial \bm D}$. However, only limited Voronoi points influence a grid element to form a high-density area as the softmax function finds the minimum distance to a given grid point and the edges are formed in areas where the minimum distances to multiple Voronoi points are similar. Therefore, the computation cost can be reduced to $O(k^2)$ if we only loop through $k$ Voronoi points closest to the grid element, and the added cost for searching k-nearest neighbor is $O(k\log N)$. Thus, we find k-nearest neighbors every time before we compute the density formed by the Voronoi points and only use the k-nearest neighbors also to compute the gradients. \gv{The sensitivity calculation is the performance bottleneck of our current implementation. As shown in Figure \ref{fig:k_statistics}, there is a linear relationship between number $k$ and computational cost of Voronoi density calculation (Algorithm \ref{alg:density}) and a quadratic relationship between number $k$ and Voronoi gradients calculation (Algorithm \ref{alg:dc_dx}), with significantly more time cost on the gradient calculation part. We use $k$ in the range from 7 to 20 for numerical examples to accommodate examples with high anisotropy.}

\paragraph{Workflow}
Overall our method has a similar optimization procedure as that of SIMP. However, our method uses Voronoi design variables $\bm D_m$ and $\bm x_m$ instead. We solve the continuous optimization problem defined by (\ref{eq:topo}) using a gradient-based, iterative optimization scheme. Four major steps are computed in each iterative step. First, we need to calculate the density $\rho$ based on the positions $\bm x_m$ and the metric tensors $\bm D_m$, as shown in Algorithm \ref{alg:density}. Second, we perform the FEM solve using an MGPCG (multi-grid preconditioned conjugate gradient) solver to obtain displacement $\bm u$ of each element and calculate the corresponding compliance $c$. Third, we compute the gradients of compliance $c$ and volume $V$ with respect to $\bm D_m$ and $\bm x_m$. \gv{For clarification, we give the pseudo-code of calculating $\frac{\partial c}{\partial \bm x}$ using neighbor search in Algorithm \ref{alg:dc_dx}, the other gradients are calculated similarly.} Last, we update the design variables using the method of moving asymptotes method (MMA) \cite{Svanberg1987TheMO,dumas2018mma}. The pseudo-code of our overall workflow is given in Algorithm \ref{alg:overall}.

%% file: results.tex
\section{Numerical results}
\label{sec:results}

We discuss the numerical results obtained using our approach. For all the examples, we assume isotropic elasticity specified by Young's modulus $E=1$ and Poisson's ratio $\mu=0.3$. In the optimization, the lower bound of Young's modulus $E_{\mathrm{min}}$ is $1e-9$. 
We use $\beta=50$ for the Voronoi representation.  \rv{We conduct all of our numerical examples on a regular lattice grid with square elements in 2D and cube elements in 3D. Ideally, we assume that we have an unlimited resolution of the background grid. However, with the limited computational resource, the finite element size should be at most half of the thickness of Voronoi edges.} Our experiments consist of four parts, including the ablation tests (Section \ref{sec:results_verification}), standard tests (Section \ref{sec:results_standard}), biomimetic applications (Section \ref{sec:results_bio}), free-boundary examples (Section \ref{sec:results_free_boundary}), and three-dimensional results (Section \ref{sec:results_3d}).



\subsection{Ablation tests}
\label{sec:results_verification}

\begin{figure}
     \centering
     \begin{subfigure}[b]{0.3\textwidth}
         \centering
         \includegraphics[width=\textwidth]{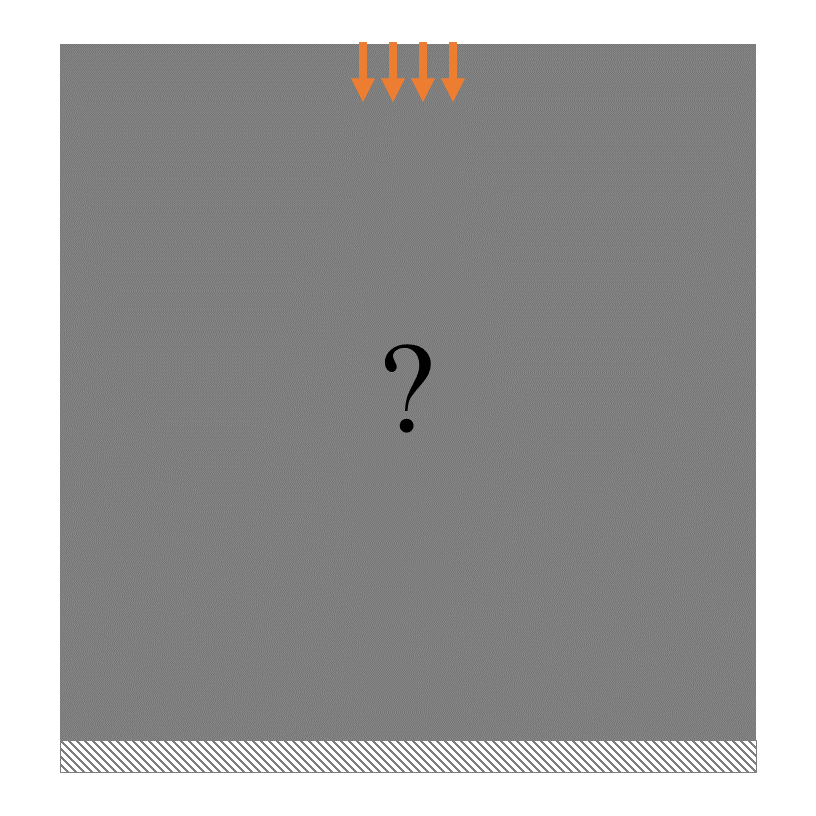}
         \caption{}
         \label{fig:push_down_setup}
     \end{subfigure}
     \hfill
     \begin{subfigure}[b]{0.3\textwidth}
         \centering
         \includegraphics[width=\textwidth]{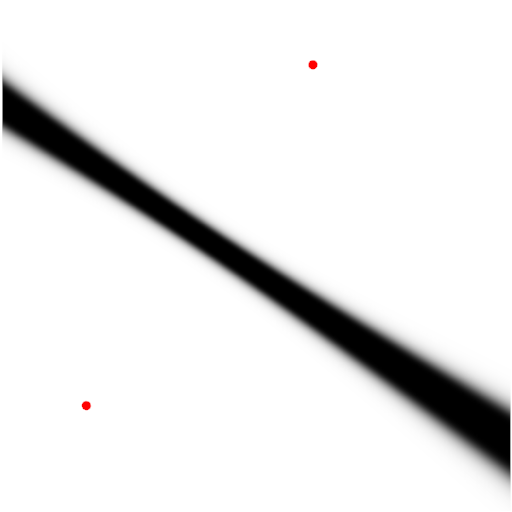}
         \caption{}
         \label{fig:push_down_0}
     \end{subfigure}
     \hfill
     \begin{subfigure}[b]{0.3\textwidth}
         \centering
         \includegraphics[width=\textwidth]{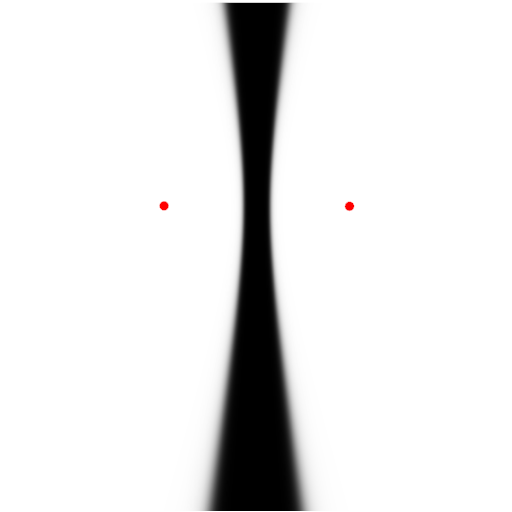}
         \caption{}
         \label{fig:push_down_122}
    \end{subfigure}
    \hfill
    \caption{(a) Setup for the pushdown experiment. (b) Initial density distribution. (c) Density distribution after optimization. The red points indicate the positions of two Voronoi points.}
    \label{fig:push_down}
\end{figure}

\begin{figure}
     \centering
     \begin{subfigure}{0.3\textwidth}
         \centering
         \includegraphics[width=\textwidth]{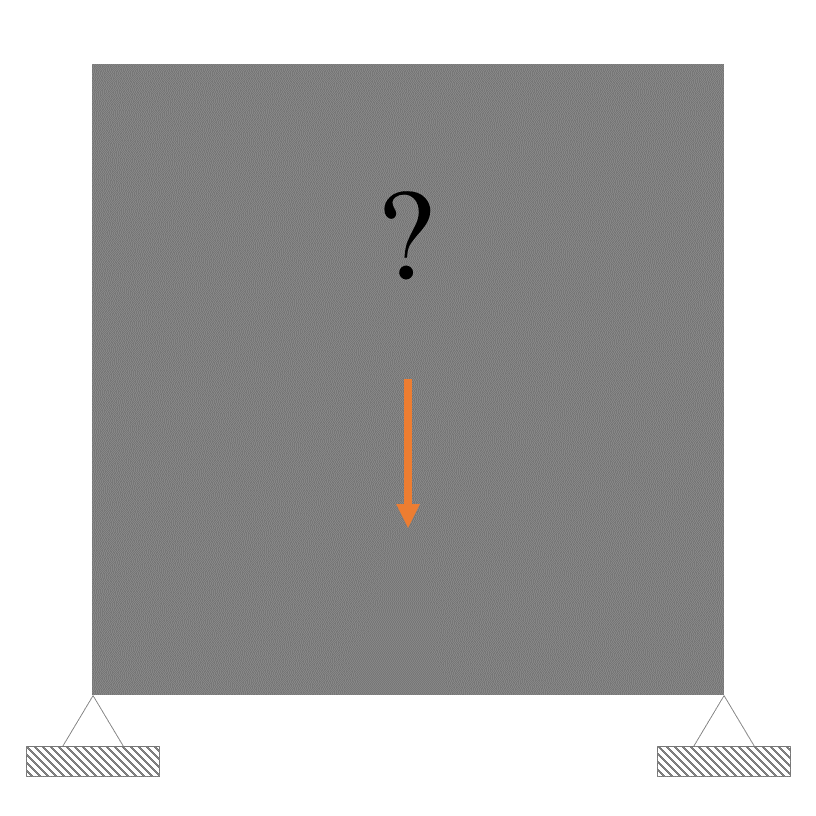}
         \caption{}
         \label{fig:middle_setup}
     \end{subfigure}
     \hfill
     \begin{subfigure}{0.3\textwidth}
         \centering
         \includegraphics[width=\textwidth]{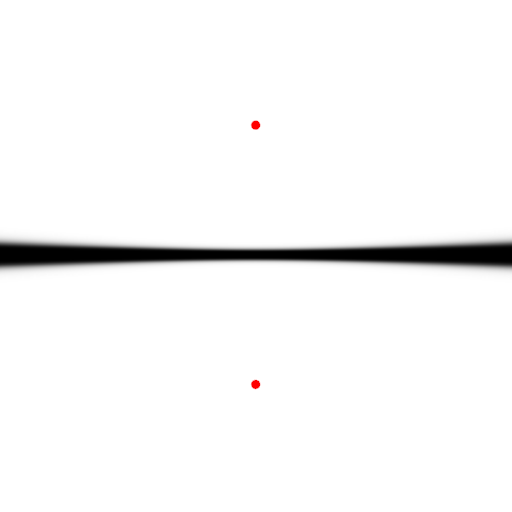}
         \caption{}
         \label{fig:middle_0}
     \end{subfigure}
     \hfill
     \begin{subfigure}{0.3\textwidth}
         \centering
         \includegraphics[width=\textwidth]{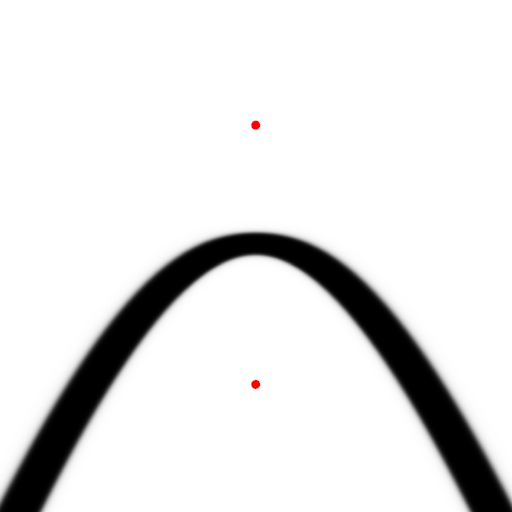}
         \caption{}
         \label{fig:middle_249}
    \end{subfigure}
    \hfill
    \caption{(a) Setup for the middle drag experiment. (b) Initial density distribution. (c) Density distribution after optimization. The red points indicate the positions of two Voronoi points.}
    \label{fig:middle}
\end{figure}

We first conduct two ablation tests to verify the roles of position $\bm x_m$ and metric tensor $\bm D_m$ separately. Both experiments were carried out with two Voronoi points on a $512\times512$ grid.
%
We first test the optimization of $\bm x_m$ (see Figure \ref{fig:push_down}). 
The boundary condition is set as a distributive force within a narrow region on the top and the bottom nodes being fixed. We initialize the position of two Voronoi points randomly as shown in Figure (\ref{fig:push_down_0}), where a tilted beam is formed, not connected with any fixed nodes on the bottom. The metric tensor $D$ for each Voronoi point is set to be $80$ multiplying an identity to ensure an appropriate edge thickness. 
After convergence, as shown in Figure (\ref{fig:push_down_122}), the beam is rotated to the middle of the domain by the moving control points, connecting the force load and the fixed base to minimize the structure's compliance.
We set up a similar optimization problem as in Figure \ref{fig:middle_setup} to verify the role of $\bm D_m$. The bottom-left and the bottom-right corners are fixed and an external load is applied in the center. We initialize two Voronoi points with their position fixed and only optimize for their metric tensor.
The program obtains $\bm D_1=[282, 1.4e-4; 1.4e-4, 89]$ for the lower point and $\bm D_2=[228, -1.3e-4; -1.3e-4, 106]$ for the higher point after convergence. The optimized design forms an arch to minimize compliance while satisfying the pre-defined target volume constraints.
The two tests jointly show that both design variables are effective in guiding the structure to evolve to a low-compliance status.

\subsection{Standard tests}
\label{sec:results_standard}
We optimize $\bm x$ and $\bm D$ for cellular cantilever beam structures with minimum compliance. 
\begin{figure}
    \center
    \begin{subfigure}{0.4\textwidth}
        \center
        \includegraphics[width=\textwidth]{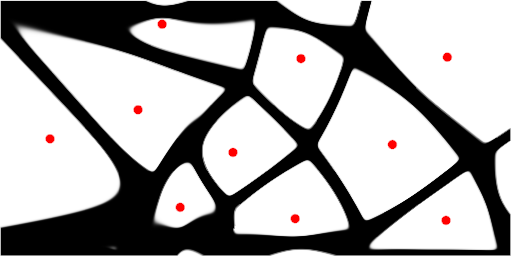}
        \caption{}
        \label{fig:cb_v}
    \end{subfigure}
    \hfill
    \begin{subfigure}{0.4\textwidth}
        \center
        \includegraphics[width=\textwidth]{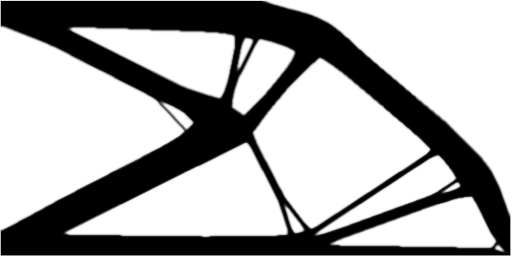}
        \caption{}
        \label{fig:cb_SIMP}
    \end{subfigure}
    \hfill\\
    \begin{subfigure}{0.7\textwidth}
         \includegraphics[width=\textwidth]{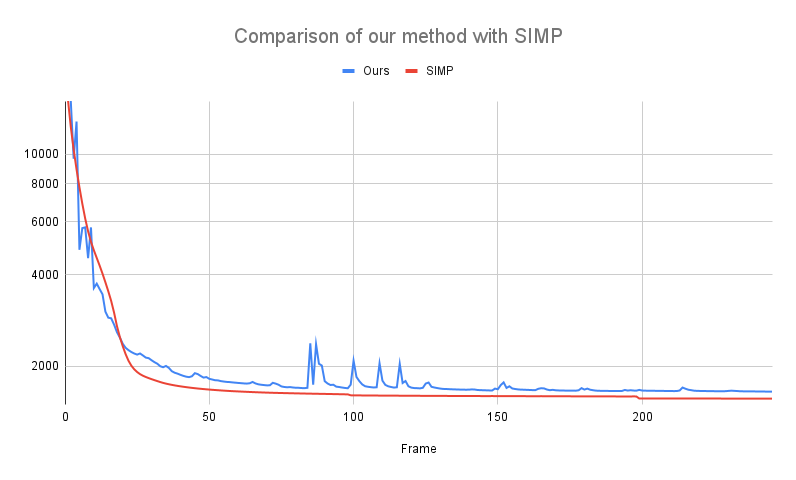}
         \caption{}
         \label{fig:cb_comp_plot}
    \end{subfigure}
    \caption{Optimization results of two-dimensional cantilever beams. (a) Our result. (b) Result obtained using SIMP. (c) The objective convergence comparison.}
    \label{fig:cb_comp}
\end{figure}
The optimized design after 245 iterations is shown in Figure \ref{fig:cb_comp}(a). We first randomly initialize 18 points in a $6\times3$ coarse grid of a $512\times256$ domain with an initial $\bm D=[150,0\rv{;~}0,150]$. Our optimized structure resembles that of a standard topology optimization result, Figure \ref{fig:cb_comp}(b), as the long beams are formed both on the bottom and top, connecting the fixed boundary with the loads. We did not visualize Voronoi points outside the domain, which also contributes to the structure. 
The compliance of our structure is $5.5\%$ above the one obtained by the standard method. 
We conjecture that the slightly larger compliance is in part due to the waste materials formed on the top right corner, which is not contributing to lower compliance and the constrained design space for cellular patterns. 
We show a comparable optimization convergence performance to the standard SIMP method (see Figure~(\ref{fig:cb_comp_plot})), in which both methods converge within 150 frames. \rv{We conjecture the spikes in the objective curve to be caused by a relatively significant change in density distribution when the positions/metric tensors change mildly in value. These jitters can be mitigated by further decreasing the move limits of the optimization parameters.}

\begin{figure}
    \centering
    \includegraphics[width=0.49\textwidth]{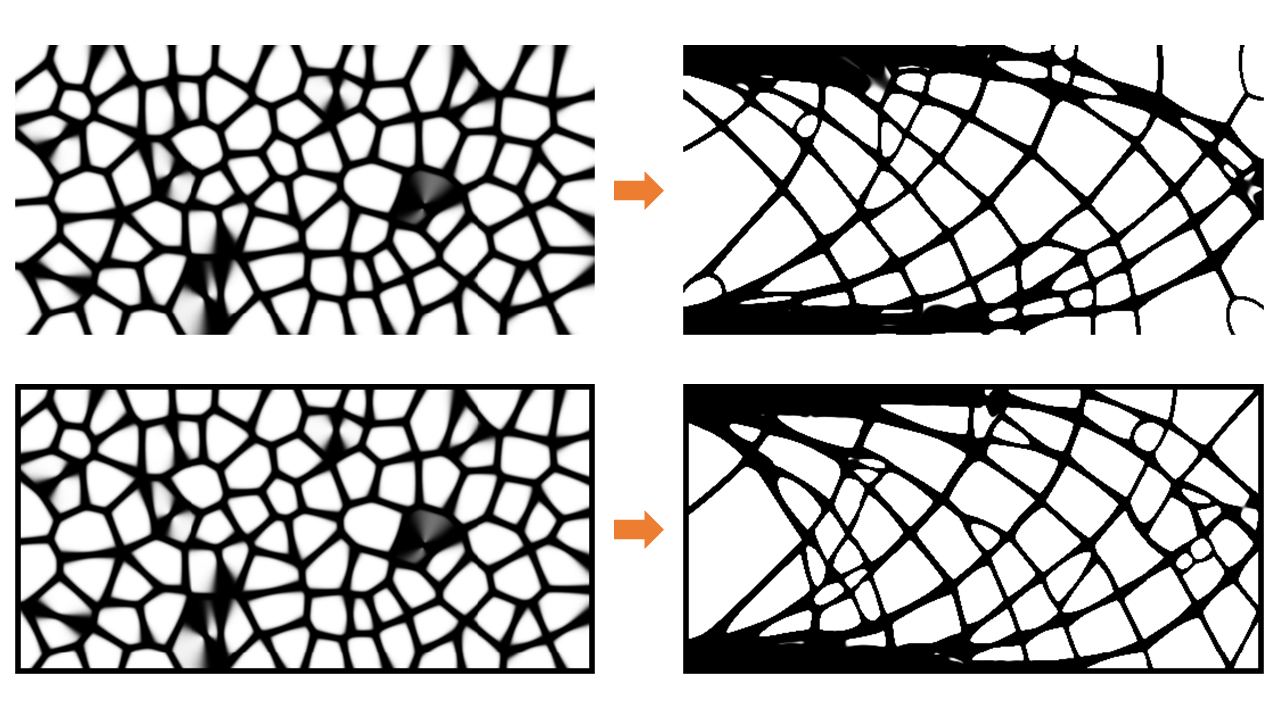}
    \includegraphics[width=0.49\textwidth]{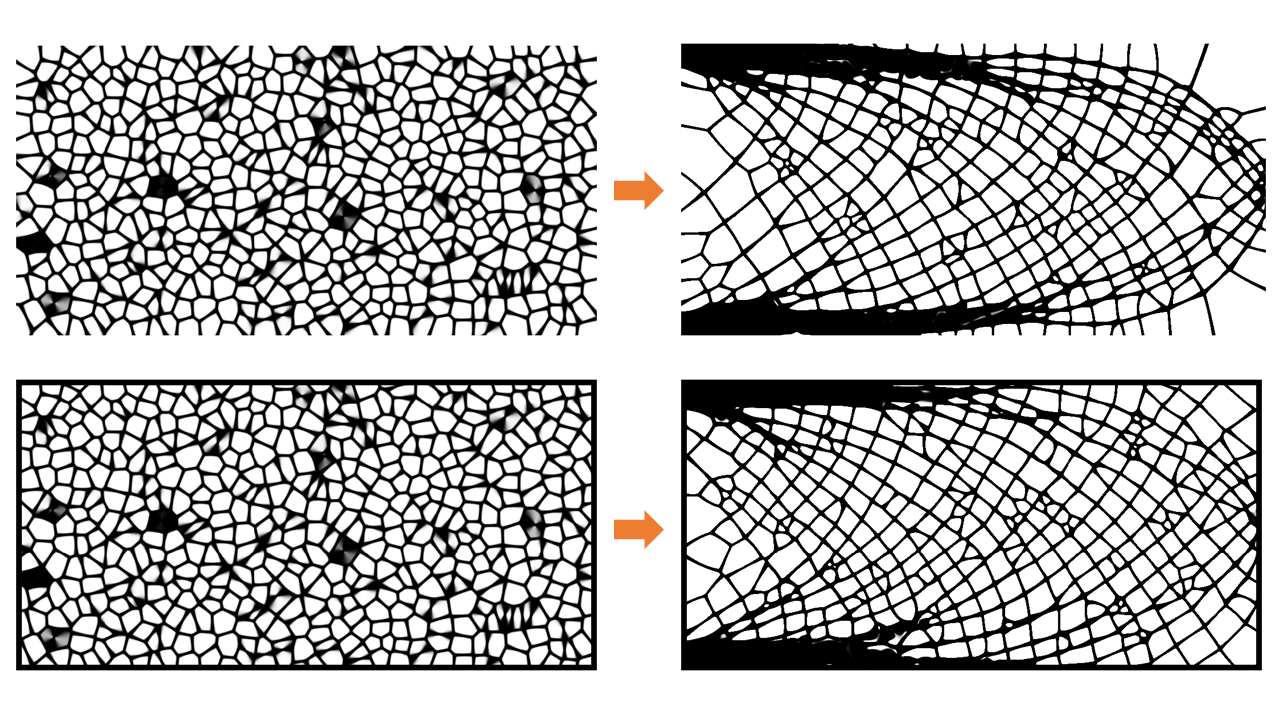}
    \caption{\rv{Topology optimization of Voronoi structures with distributed force load on the right side. The top row shows the results without density on the boundary while the bottom row shows the result that has one layer of density on the boundary. The left side is the result formed by 98 Voronoi points with resolution $512\times256$ while the results on the right side are formed by 512 Voronoi points with a grid resolution of $1024\times512$.}}
    \label{fig:cb_frame}
\end{figure}

\begin{figure}
    \centering
    \includegraphics[width=0.8\textwidth]{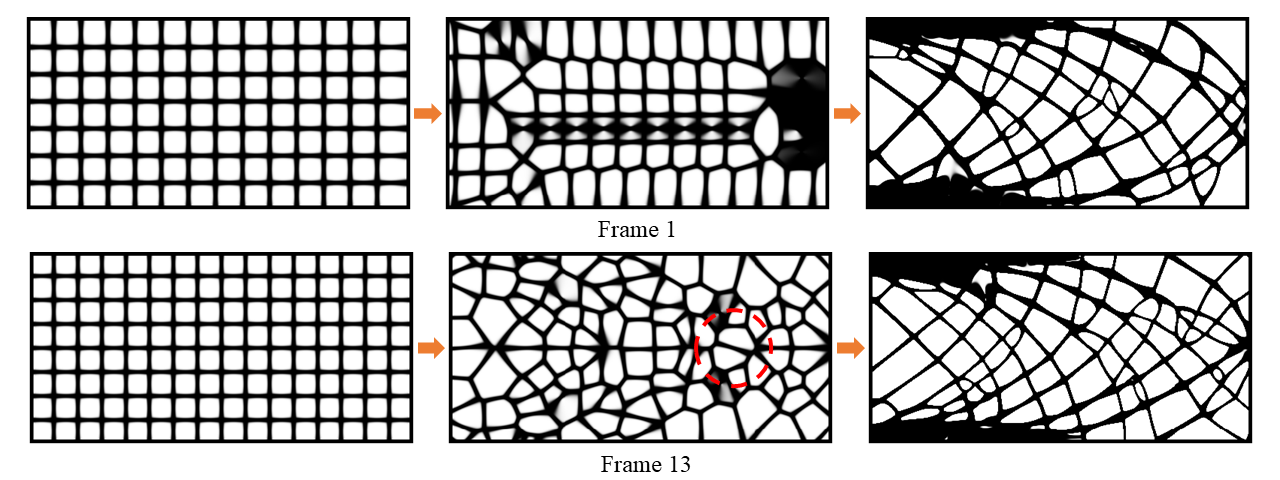}
    
    \caption{\rv{Optimization result of the framed cantilever beam with $14\times7$ and $16\times8$ points respectively.}}
    \label{fig:cb_frame_regular}
\end{figure}

The second example is a framed cantilever beam with a distributed force load. The nodes on the left side are all fixed, and downward-pointing forces are exerted in the middle region on the right side, occupying the middle 20 percent of the total height. The rectangular frame of the $1\times0.5$ domain is fixed to be filled with solid materials with a beam thickness of 0.01 \rv{(in the framed cases)}. The target volume is set to be 0.35. The optimization results with the different number of initial points are shown in Figure \ref{fig:cb_frame}. From left to right, we observe that a more distributed structure is formed with more Voronoi points. \rv{All} optimized structures form axial beams that traverse from the left top corner to the right middle, similarly, from the left bottom corner to the right middle. Additionally, both form the dark ``anchor" regions on the top left and bottom left corners. These features demonstrate the spatially adaptive nature of our method, which are also the keys to minimizing compliance in this boundary setup. \rv{Comparing the upper row without density on the frame and the lower row with density on the frame, the overall structures look similar except that there is less material connecting to the upper and lower corner in the examples without frame. For the results in the $512\times 256$ grid, 17 Voronoi points move outside of the domain in both cases. For the results in $1024\times 512$ grid, 83 Voronoi points move outside of the active cell domain for the case where density fills the frame while 98 points move outside of the domain for the unframed case.
\\
In Figure \ref{fig:cb_frame_regular}, we initialize the points in the centers of the coarse grids. Since the stress distribution is symmetric across the horizontal axis, the optimization initialization of points in the $14\times7$ coarse grid exhibits an asymmetrical pattern as early as frame 1. When we optimize the structure with the initialization of points in the $16\times8$, it preserves a symmetric structure until frame 13. Also, we observe that starting from frame 13, the volume constraint is better satisfied. So we conclude that the optimizer prioritizes satisfying volume constraint than strictly following gradient from that point on.}

\begin{figure}
    \includegraphics[width=\textwidth]{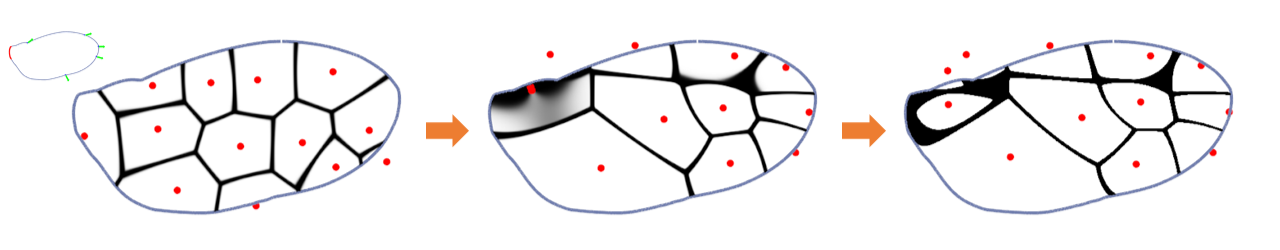}
    \caption{
    Evolution of Voronoi points and corresponding density for the Diptera wing example. From left to right at iteration 1, 100 and 299 respectively. The boundary condition is illustrated on the inset picture on the top left\rv{, where the red points indicate the fixed nodes, blue points indicate the nodes that receive forces and the green lines indicate the force directions.}}
    \label{fig:diptera}
\end{figure}

\begin{figure}
    \centering
    \includegraphics[width=0.95\textwidth]{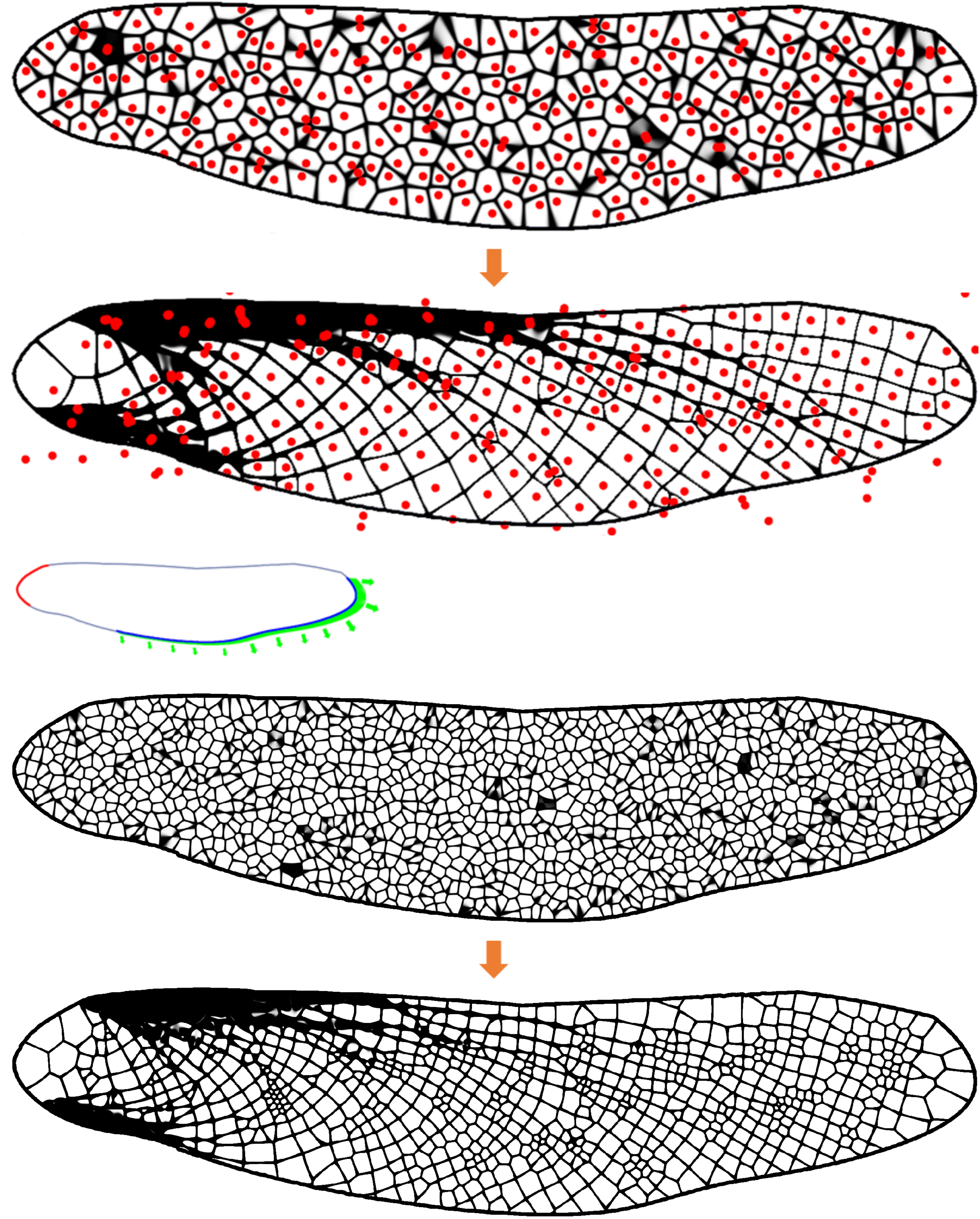}
    \caption{The evolution result of odonata front wing. The boundary condition is illustrated in the inset picture on the middle left. Top row: The optimized design with 290 Voronoi points in a $1024\times256$ domain. Bottom row: The optimized design with 1116 Voronoi points in a $2600\times650$ domain.}
    \label{fig:odonata}
\end{figure}

\begin{figure}
    \centering
    \includegraphics[width=0.8\textwidth]{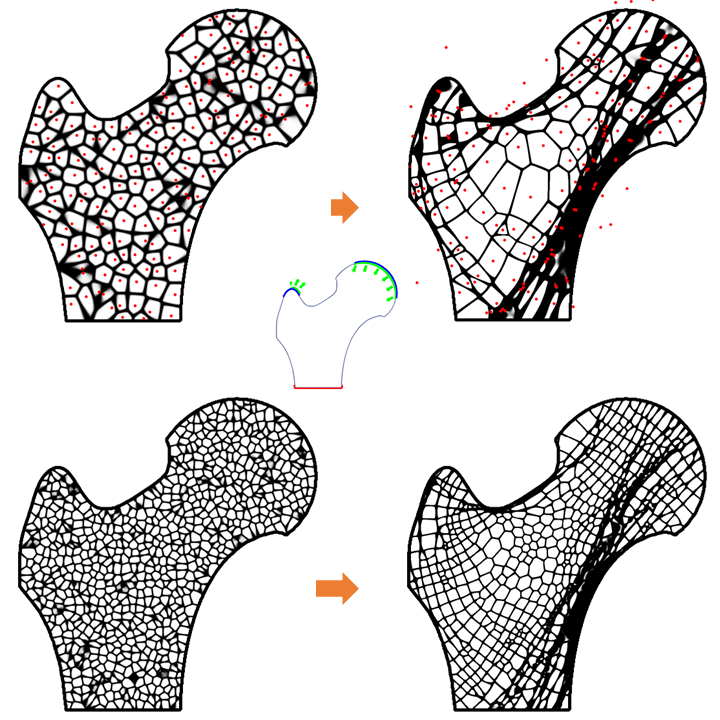}
    \caption{Evolution of cellular structures for the bone example. Middle: Boundary condition for the bone example. Similar to the Diptera example, the darker blue contour indicates the boundary of the design domain. The red nodes on the bottom are the fixed nodes. The distributed forces are applied to the two locations indicated by blue circles with green bars attached, which visualize the direction and the strength of the forces. Top row: Left is the initialization of 196 Voronoi points \rv{on a $1024\times 1024$ grid} and Right is the optimized result after 249 iterations. Bottom row: Left is the initialization of 773 Voronoi points \rv{on a $1024\times 1024$ grid} and Right is the optimized result after 249 iterations.}
    \label{fig:bone}
\end{figure}

\subsection{Biomimetic structures}
\label{sec:results_bio}
To demonstrate the efficacy of our approach in optimizing biomimetic cellular structures, we optimize the structure within a Diptera wing contour as in Figure \ref{fig:diptera}. We randomly initialize 13 Voronoi points in a $6\times3$ grid with $\bm D=[400,0\rv{;~}0,400]$.
The target volume fraction is set to 0.15. \rv{For processing the irregular boundary here, we first convert a silhouette of the Diptera wing shape into a level set field using fast marching algorithm \cite{sethian1999fast}. Then, we mark the cells with negative level set values as active and mark the cells with positive level set values as inactive cells, which we set $\rho_e = 0$ and do not optimize the density distribution. The nodes that have active neighbor cells and inactive neighbor cells are marked as boundary nodes and we apply boundary forces on them. The directions of the boundary forces are partially determined by the gradient of the level set as well.} As shown in Figure \ref{fig:diptera}, the Voronoi points are distributed in the design domain evenly and the initial Voronoi tessellation is similar to a centroidal Voronoi tessellation. As optimization proceeds, however, the Voronoi partitions are more elongated, which conforms to the force direction. Also, the beams on the left side connecting to the fixed points are relatively thicker than the beams on the right side. Note that although some regions with diffused density appear as two points closer to each other, they separate naturally in the later iterations after the stronger filter strength leads to a more binary density pattern. A denser cellular pattern in an Odonata front wing is shown in Figure \ref{fig:odonata}, where 290 points with an initial $\bm D=[1300,0\rv{;~}0,1300]$ and 1116 points with an initial $\bm D=[2600,0\rv{;~}0,2600]$ are optimized after 249 iterations. \rv{The irregular boundary is set similarly in this example except that the cells within a certain positive epsilon values are marked as boundary cells with $\rho_e=1$, and and the boundary nodes are in between boundary and inactive cells.} It can be observed that thick beams are formed towards the fixed nodes on the right side and there is a spectrum of variations in the edge thickness.

We also test our algorithm to optimize a bone's interior. As illustrated in Figure \ref{fig:bone}, we fix the bottom of the bone shape and add distributed pulling forces on the top left and distributed pushing force on the top right. The evolution of the topology optimization process is shown in Figure \ref{fig:bone} at selected frames. We can see that the overall density of distribution is reached as early as iteration 50, while a more binary and refined structure is achieved in the end. Our results are consistent with the characteristic femur bone structure with elongated cells and inhomogeneous density distribution across the entire domain. The result formed with 773 points has a more distributed density while maintaining the overall cellular directions. In this example, the structure formed with fewer points achieves a lower compliance \rv{102} than that of the structure formed by more points (115) which implies that a less even density distribution is more desirable under this specific boundary condition.

\begin{figure}
    \centering
    \includegraphics[width=0.243\textwidth]{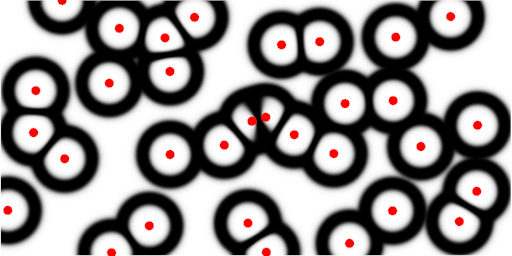}
    \includegraphics[width=0.243\textwidth]{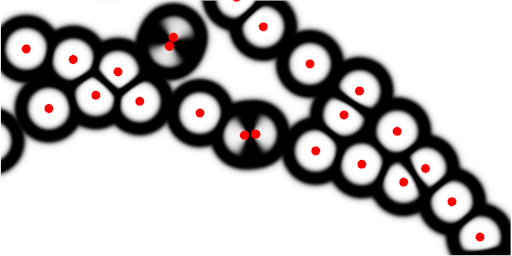}
    \includegraphics[width=0.243\textwidth]{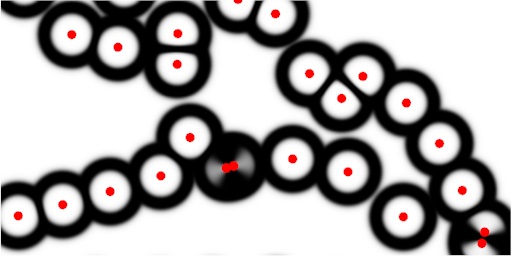}
    \includegraphics[width=0.243\textwidth]{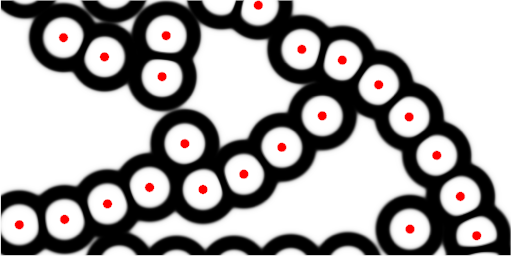}\\
    \includegraphics[width=0.243\textwidth]{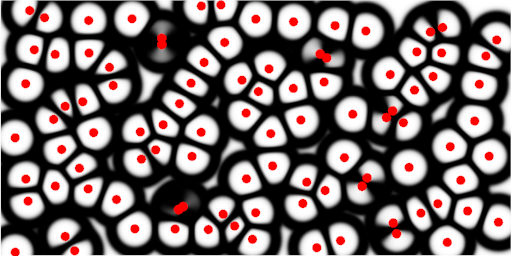}
    \includegraphics[width=0.243\textwidth]{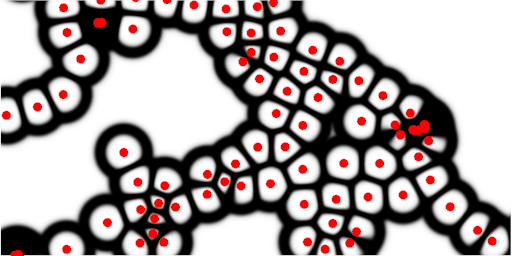}
    \includegraphics[width=0.243\textwidth]{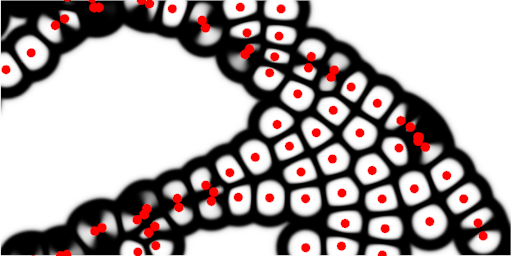}
    \includegraphics[width=0.243\textwidth]{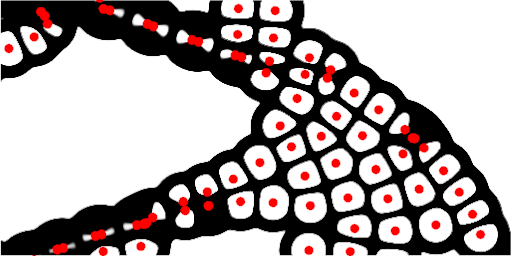}
    \caption{Cellular structure with free boundary, only optimizing x. From left to right, the columns are optimization results at iteration 0, 15, 50 and 249 respectively. The top row is the result formed by 32 Voronoi points while the results on the bottom are formed by 98 Voronoi points.}
    \label{fig:cb_f}
\end{figure}

\begin{figure}
    \centering
    \includegraphics[width=0.243\textwidth]{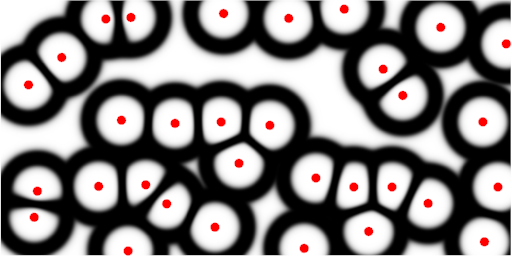}
    \includegraphics[width=0.243\textwidth]{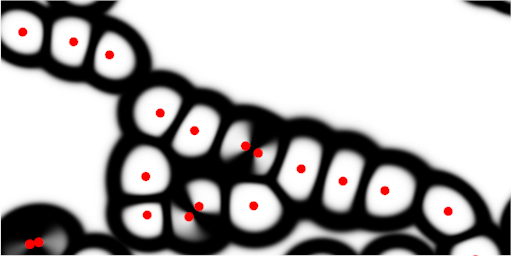}
    \includegraphics[width=0.243\textwidth]{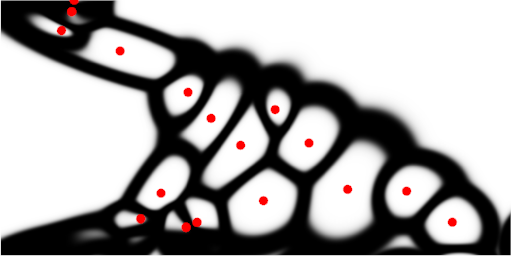}
    \includegraphics[width=0.243\textwidth]{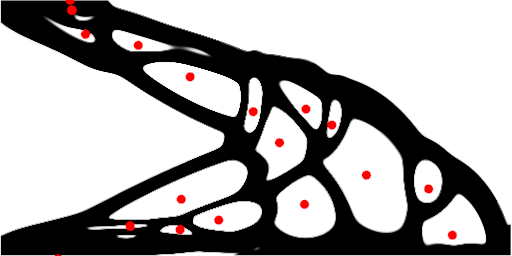}\\
    \includegraphics[width=0.243\textwidth]{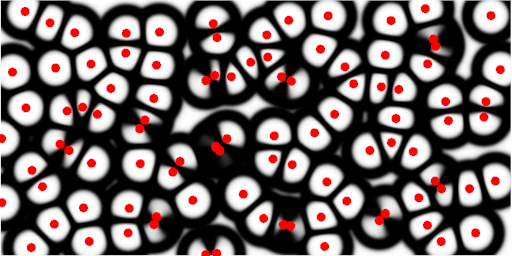}
    \includegraphics[width=0.243\textwidth]{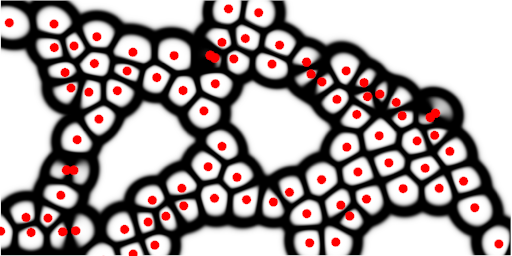}
    \includegraphics[width=0.243\textwidth]{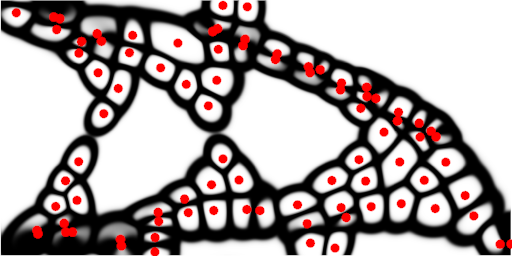}
    \includegraphics[width=0.243\textwidth]{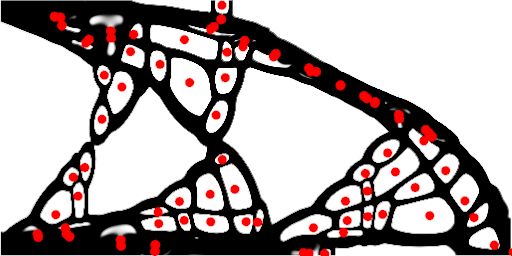}
    \caption{Anisotropic cellular structure with free boundary, optimizing both $\bm x$ and $\bm D$. From left to right, the columns are optimization results at iteration 0, 15, 50 and 249, respectively. The top row is the result formed by 32 points while the results on the bottom are formed by 98 points.}
    \label{fig:cb_f_d}
\end{figure}

\begin{figure}
    \centering
    \includegraphics[width=0.42\textwidth]{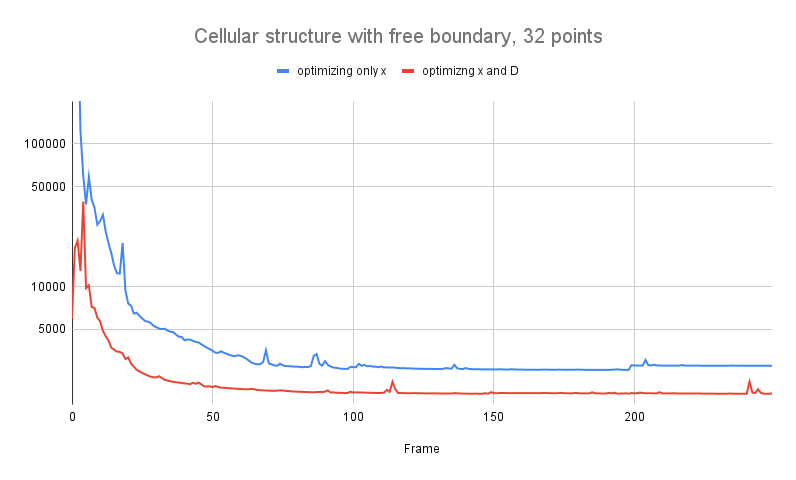}
    \includegraphics[width=0.42\textwidth]{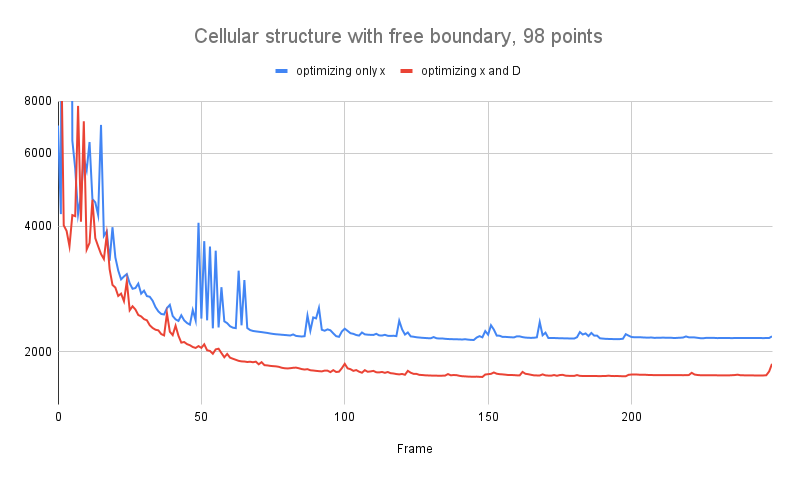}
    \caption{\rv{The comparison of compliance function curves of cellular structures with free boundary formed only optimizing $\bm x$ and optimizing both $\bm x$ and $\bm D$ with 32 and 98 Voronoi points respectively.}}
    \label{fig:cb_f_comp}
\end{figure}

\subsection{Foam structures with free boundary}
\label{sec:results_free_boundary}
Our method is also capable of optimizing foam structures with free boundaries, ergo, no need to connect cell walls all the way to the border of the design domain, thus forming cellular structures seen in nature such as cell clusters. In this subsection, we demonstrate the unprecedented cellular structure with/without enabling the anisotropic feature. All experiments are done with the standard Cantilever Beam boundary condition.

In Figure \ref{fig:cb_f}, we optimize only the positions of the Voronoi control points. On a $512\times256$ grid, we initialize 32 points and 98 points randomly in coarse grids $8\times4$ and $14\times7$, respectively. The initial $\bm D$ is set to $[250,0\rv{;~}0,250]$ for the experiment with fewer points, and $[350,0\rv{;~}0,350]$ for the experiment with more points. The value of $\epsilon_s$ is set to $1e-7$ to obtain an appropriate radius for the cell size. In the beginning, the cells spread out the entire domain, without a structure. After 15 iterations, cells start to cluster in order to form a long beam connecting the left fixed side with the force load on the right bottom corner. A more clear 2-beam structure forms at the 50th iteration, and final structures that more resemble a standard cantilever beam form in the end. Notice that since there is no structure connecting the force load and the fixed nodes in the beginning, the FEM solve does not converge. However, the optimization process is not obstructed by this issue.

With the same initial setup, but enabling the optimization of $\bm D$, we can obtain cellular structures with morphed cells and lower compliance \rv{as shown in Figure \ref{fig:cb_f_d}}. Similarly, the cells are dispersed randomly in a coarse grid at the start point and start to form beam structures with less distorted cells at iteration 15. The structure at iteration 50 consists of more distorted cells while maintaining the overall structure. In the end, thicker beams are formed on the bottom and across from the left-top to the bottom right corner for both cases.

\rv{In Figure \ref{fig:cb_f_comp}, we compare the compliance function curves of the cellular structures with free boundary formed only optimizing $\bm x$ and optimizing both $\bm x$ and $\bm D$ with 32 and 98 Voronoi points respectively. In both cases, the structures obtained by optimizing both $\bm x$ and $\bm D$ reach much lower compliance values.}

\rv{Furthermore, we compare the compliance of the structure with free boundaries and without free boundaries as shown in \ref{fig:cb_f_comp_2d}. In this case, the final compliance of the structure with free boundaries is higher than that without free boundaries, even though the structure with the free boundary removes the extra material in the upper right corner that doesn't contribute to lowering the compliance. This is the constraint of the free boundary structure as the cells need to form a closed cellular structure and there are no connecting cells with the fixed nodes on the left side under this volume constraint.}

\begin{figure}
    \centering
    \includegraphics[width=0.33\textwidth]{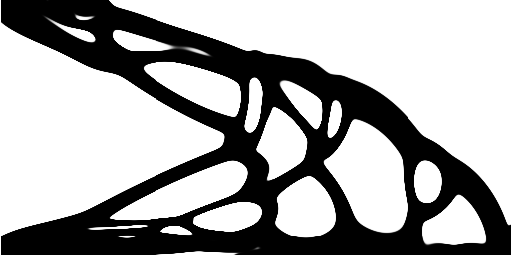}
    \includegraphics[width=0.33\textwidth]{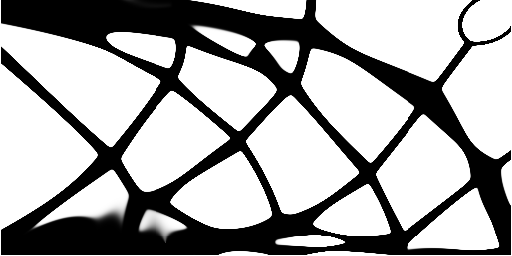}
    \includegraphics[width=0.33\textwidth]{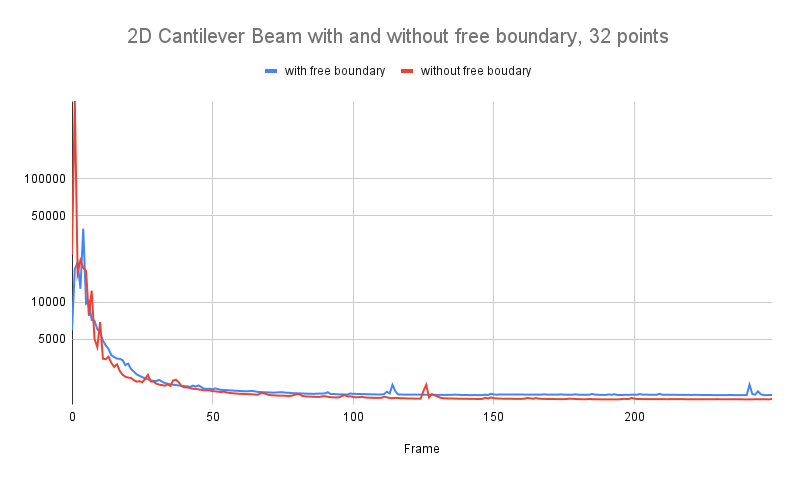}
    \caption{\rv{Comparison of Cantilever Beam with and without free boundary in 2D with 32 Voronoi points.}}
    \label{fig:cb_f_comp_2d}
\end{figure}

\subsection{3D result and performance}
\label{sec:results_3d}
\begin{figure}[ht]
    \centering
    \includegraphics[width=0.8\textwidth]{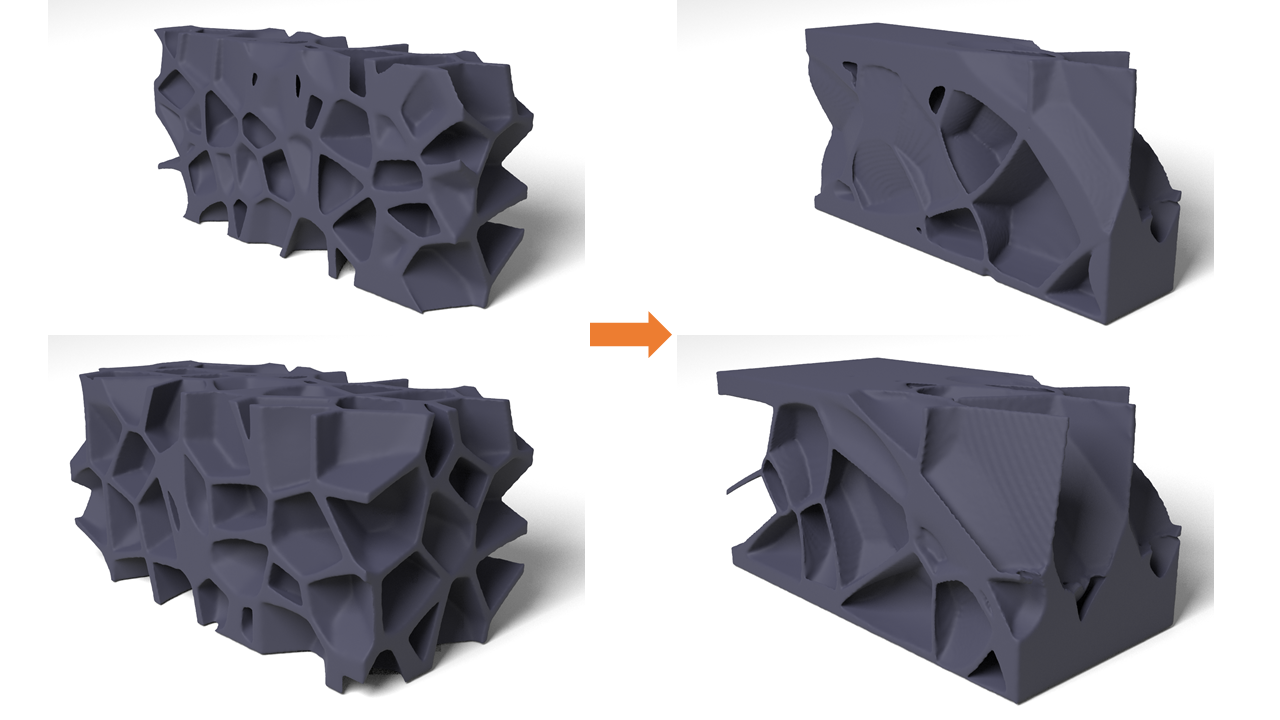}
    \caption{Rendered 3D Cantilever beam results. Left column is the structure upon initialization, right column is the optimized structure at iteration 150. The top row is the cross-sectional view. }
    \label{fig:cb_3d}
\end{figure}

\begin{figure}[ht]
    \centering
    \includegraphics[width=0.8\textwidth]{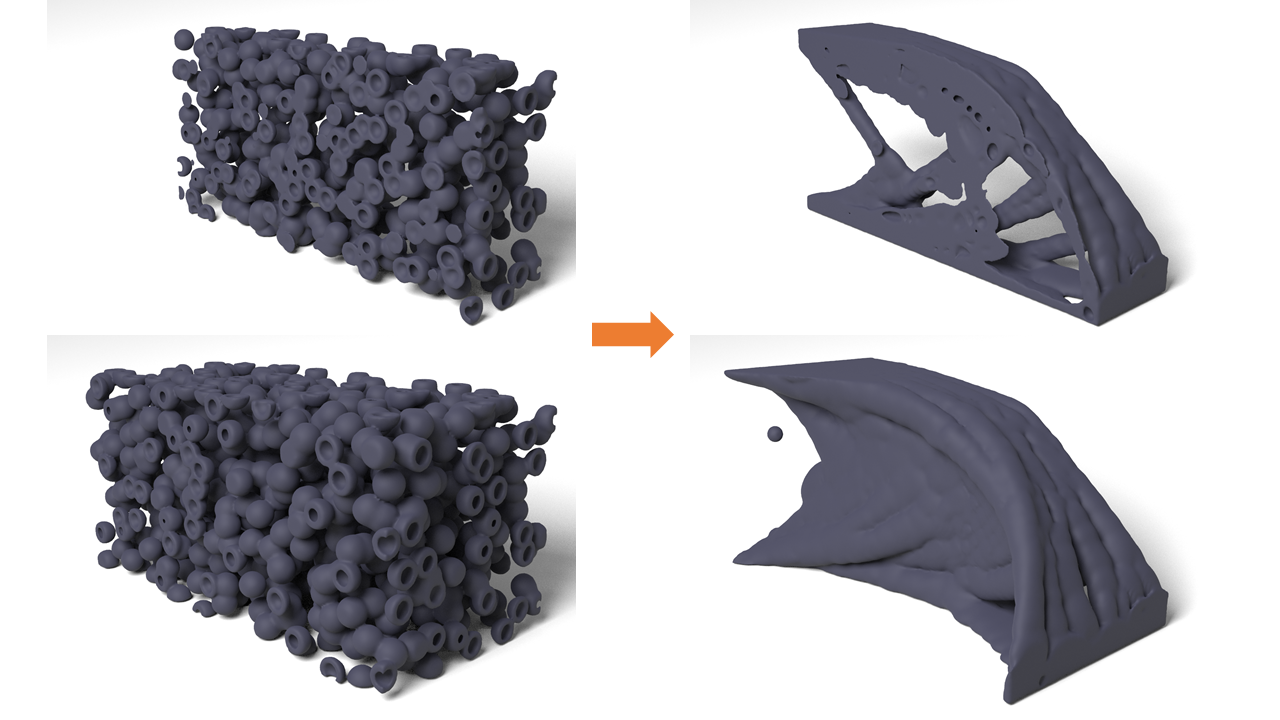}
    \caption{Rendered 3D Cantilever beam with free boundary results. Left is the structure upon initialization, right is the optimized structure at iteration 100. The top row is the cross-sectional view.}
    \label{fig:cb_3d_f}
\end{figure}

\begin{figure}
    \centering
    \includegraphics[width=0.33\textwidth]{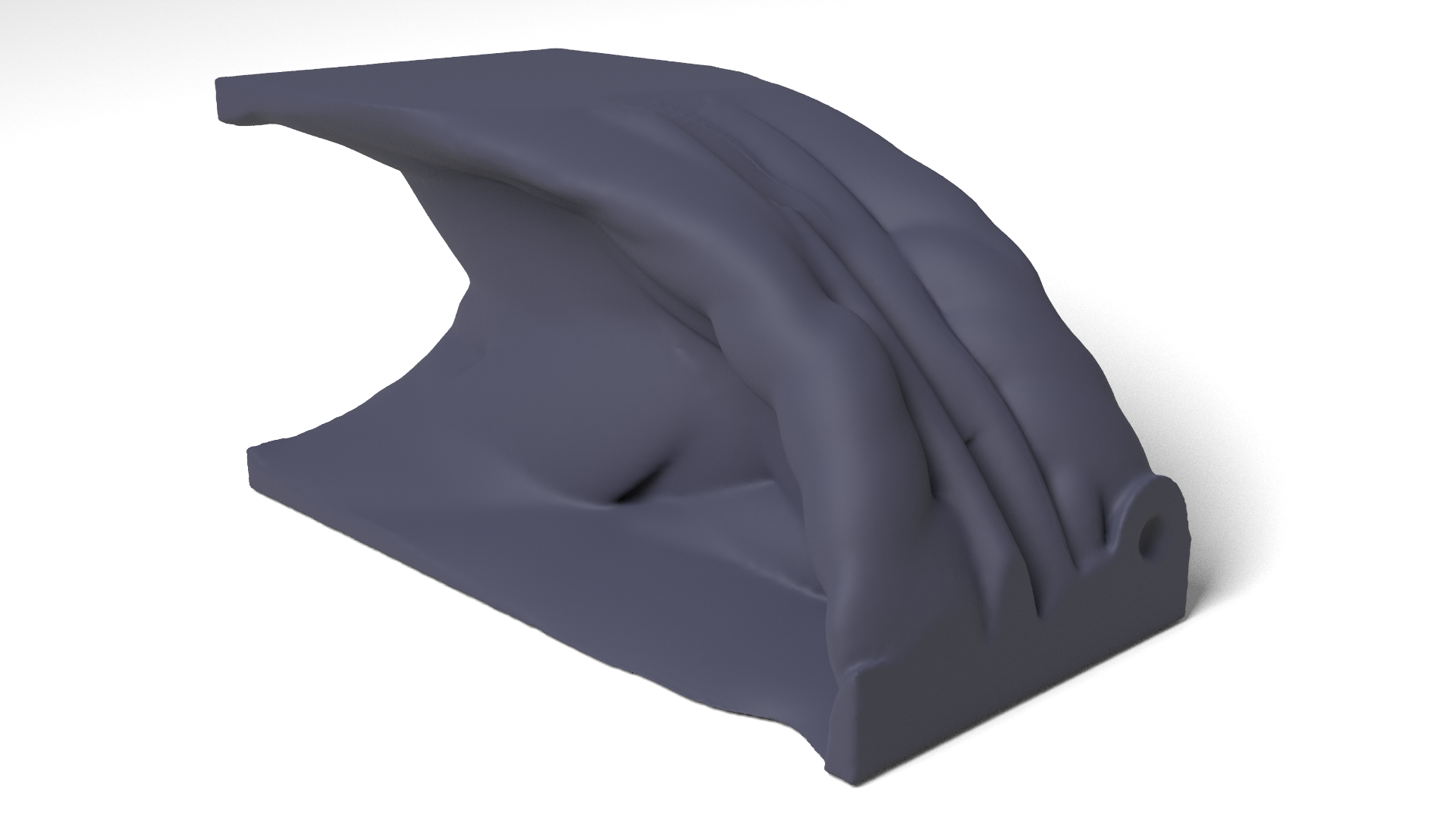}
    \includegraphics[width=0.33\textwidth]{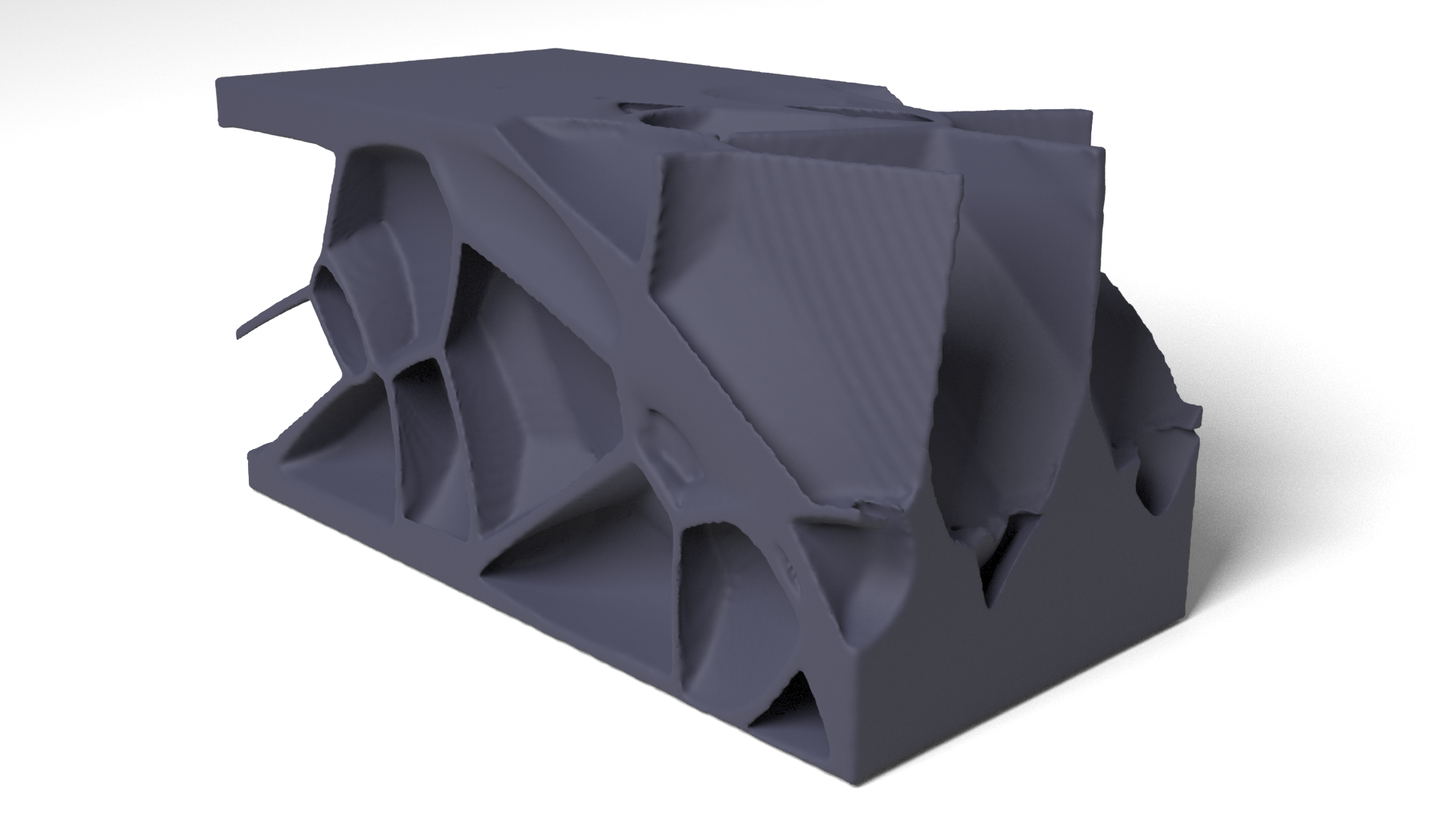}
    \includegraphics[width=0.33\textwidth]{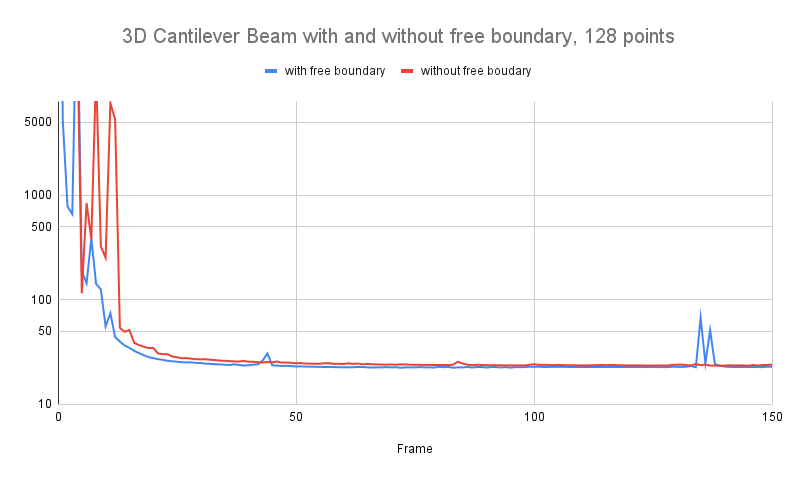}
    \caption{\rv{Comparison of Cantilever Beam with and without free boundary in 3D with 98 points.}}
    \label{fig:cb_f_comp_3d}
\end{figure}

Our method can be easily extended to 3D. Therefore we also demonstrate the 3D cantilever beam in Figure \ref{fig:cb_3d}. Here, we initialized 128 points inside a $256\times128\times128$ grid with an initial $\bm D=300I$. Large walls are formed on the bottom and across the top left to the bottom right, connecting nodes receiving force load with the fixed nodes on the left side. We use the isosurface of the optimized density field and render the images in Houdini. We can also add a free boundary for this example in 3D space, as shown in Figure \ref{fig:cb_3d_f}. 1024 points are randomly initialized in the grid with initial $\bm D=700I$. Three-dimensional bubbles are formed and stretched after optimization, and there is no wasted material in the top right corner. \rv{Similarly, we compare the resulting in compliance for the optimized 3D cellular structures with and without free boundaries, as shown in Figure. \ref{fig:cb_f_comp_3d}. For similar reasons, the compliance of the structure with free boundaries is not that much different from that without free boundaries.
}

The statistics of our high-resolution examples are summarized as in Table \ref{table:statistics}. These examples are all run on a computer with an Intel i9 9980 CPU, 18 cores, 4GPU NVIDIA GeForce RTX2080, and 128G memory. The convergence criteria given in the table is $|c_i+c_{i-1}-c_{i-2}-c_{i-3}|/(c_{i-2}+c_{i-3})<0.1\%$, where $i$ is  the iteration number. \gv{The forward part includes the time cost of updating the density distribution based on the current variables of the Voronoi structure and the FEM solve for node displacement. The gradient calculation is taking the most time because our algorithm needs to iterate through each cell to calculate the derivative of compliance to each optimizable variable.}

\begin{table}[h]
\centering
 \begin{tabular}{p{0.2\linewidth} c c c c c c} 
 \toprule
  \textbf{Example} & \textbf{Resolution} &\multicolumn{3}{c}{\textbf{Computation Time (s)}} & \textbf{Convergence} & \gv{\textbf{Neighbor Number $k$}}\\
\cmidrule{3-5}
&& \gv{Forward} & Gradient & MMA\\ 
\midrule
 framed cantilever beam & $1024\times512$& 0.721 & 1.577 & 0.006 & 70 & \gv{10}\\
 \midrule
 odonata wing &$1024\times256$& 0.511 & 1.429 &0.004 & 17 & \gv{15}\\
  \midrule
 femur bone &$1024\times1024$& 1.164 & 4.915 &0.001 & 49 & \gv{15}\\
  \midrule
 3D cantilever beam  &$256\times128\times128$& 2.857 & 20.18 &0.004 &39 & \gv{15}\\
  \midrule
 3D cantilever beam with free boundary &$256\times128\times128$& 4.65 & 80.86 & 0.021 & 59 & \gv{15}\\  
 \bottomrule
 \end{tabular}
 \caption{Statistics of high-resolution examples.}
 \label{table:statistics}
\end{table}

\subsection{Discussion}
\rv{
Compared with the standard SIMP method, our topology optimization algorithm with differentiable Voronoi diagrams manifest its computational merits in several aspects. First, the algorithm encodes a strong geometric prior into the optimizer to coerce the emergence of cellular structures, providing a novel design space that was impractical to explore using a conventional density-based method. Second, we propose a free-boundary treatment to overcome the inherent limitation of Voronoi as a domain partition structure. This implicit, free-boundary representation extends our algorithm from designing the interior part only to handling the optimization of both the boundary shape and the internal structures. Third, the implicit, differentiable nature of our Voronoi diagram provides a practical way to calculate the design sensitivities. Last, regarding the computational performance, the locality of our hybrid grid-particle representation avoids a global search for point neighbors, facilitating large-scale optimizations on a high-resolution FEM grid. 
}

\rv{
Our approach also manifest several weaknesses compared with the standard SIMP method. First, as shown in the objective curves (e.g., Figure~\ref{fig:cb_comp_plot}), our approach under-performs in terms of its compliance performance when comparing with SIMP. Because of the additional Voronoi priors employed on top of the density representation, the results tend to exhibit an ensemble of cellular features of thin membranes and small cells. These features were not favored by a SIMP optimizer, which, in contrast, tend to generate thick beams and large volumes. Second, the free-boundary representation limits the search of structures with optimal compliance too. As shown in Figure~\ref{fig:cb_f_comp_2d} and Figure~\ref{fig:cb_f_comp_3d}, the results with a free-boundary treatment, despite its removal of wasted features in the top-right corner, did not exceed the one without a free boundary (and therefore with wasted features). Third, the computational cost of the implicit Voronoi, in particular the sensitivity $d\rho/d \bm x$, though accelerated by its local search strategy, is highly dependent on the number of neighboring points within the range. A large number of neighboring points, which is necessary to accommodate highly anisotropic cells in some examples (e.g., Figures~\ref{fig:bone} and \ref{fig:cb_3d}), leads to a quadratic increase of the computation time for the sensitivity calculation. Devising neighbor-pruning strategies to reduce the soft-max sum calculation for redundant neighbors (i.e., the points not contributing to a cell's density), or devising anisotropic Kd-tree search algorithms, could be an interesting direction to explore. 
}

%% file: conclusion.tex
\section{Conclusion}
\label{sec:discussion}
We proposed a topology optimization algorithm for generalized cellular structures. A Voronoi representation that is differentiable, generalized, and encompassing a hybrid discretization is devised to construct an efficient gradient-based optimization framework. The numerical examples demonstrate that the proposed method is effective in generating organic cellular structures resembling those in nature that also minimize structural compliance. Compared to the previous literature in cellular topology optimization, our method enabled, for the first time, a differentiable representation of the previously discrete Voronoi representation and successfully incorporated this representation in a density-based standard topology optimization pipeline. Compared to the previous hybrid Eulerian-Lagrangian topology optimization methods such as \cite{guo2014doing,norato2015geometry,li2021lagrangian}, which relied on an explicit or smoothed Lagrangian representation that can move in an Eulerian domain, our approach encodes the optimization variables into an implicitly defined distance metric tensor field carried on moving particles to express complex cellular structures. The differentiable nature of our geometric representation, in conjunction with its generalized Voronoi cells based on an optimizable, non-Euclidean distance field, allows the optimization of a broad range of organic cellular structures based on first principles.  

Still, there are some limitations to our method. First, due to the definition of our differentiable Voronoi representation, the edge formed by two control points is thinner in the middle and thicker farther away from the two points. We anticipate solving this problem by exploring different forms of metric tensors to obtain a controllable and consistent edge thickness over space. 
Second, the FEM discretization in our method relies on a high-resolution grid on the background, which requires a fine grid discretization to represent a complex structure with many cells. Exploring more efficient Eulerian-Lagrangian representations, in particular, those with sub-cell discretizations, would further reduce the computational cost of the entire optimization pipeline.
Third, our current approach handles the free boundary with a fixed $\epsilon_s$ value, implying a constant and non-optimizable radius for the foamy boundary. Incorporating $\epsilon_s$ into the optimization framework can further enhance the expressiveness of the current model. 
\rv{Fourth, we optimize the wing-shaped structures only in two-dimensional spaces to demonstrate the expressiveness of our proposed differentiable cellular structure. To make a comparison with the Voronoi structures of insect wings in nature, we plan to use our method on shell finite elements.}
We also consider future work to investigate multi-level cellular structures, featuring cellular structures at multiple length scales with a parent-child relationship, and extend the Voronoi structural representation to other structural optimization problems with different objectives and physical constraints, motivated by the nature of various insect wings that are driven by different biological or mechanical models.